\documentclass[a4paper]{article}
\pdfoutput=1
\usepackage{jhepmod}

\usepackage[utf8]{inputenc}

\usepackage[table ]{ xcolor}
\usepackage{enumitem}
\usepackage{multirow}

\usepackage{amsmath}
\usepackage{amsthm}

\usepackage{tikz}
\usetikzlibrary{shapes,arrows}
\tikzstyle{startstop} = [rectangle,rounded corners, minimum width=3cm,minimum height=1cm,text centered, draw=black,fill=red!30]
\tikzstyle{io} = [trapezium, trapezium left angle = 70,trapezium right angle=110,minimum width=3cm,minimum height=1cm,text centered,draw=black,fill=blue!30]
\tikzstyle{process} = [rectangle,minimum width=3cm,minimum height=1cm,text centered,text width =3cm,draw=black,fill=orange!30]
\tikzstyle{decision} = [diamond,minimum width=3cm,minimum height=1cm,shape aspect=3,inner sep = 0.4pt,text centered,draw=black,fill=green!30]
\tikzstyle{arrow} = [thick,->,>=stealth]
\tikzstyle{shadow}=[preaction={fill=black,opacity=.5,transform canvas={xshift=0.5mm,yshift=-0.5mm},shading=radial,shading angle=20},fill=red]

\tikzstyle{ellipse}=[draw, rectangle, minimum width=2.8em, rounded corners=6pt,line width=0.5pt]
\tikzstyle{pxsbx}=[trapezium, trapezium left angle=75, trapezium right angle=105, minimum width=3em, text centered, draw = black, fill=white,line width=0.5pt] 
\tikzstyle{lingxing}=[draw,diamond,shape aspect=3,inner sep = 0.4pt,thick,font=\itshape,line width=0.5pt]

\usepackage{amssymb}
\usepackage{graphicx,color}

\usepackage{amsmath,amsthm}

\usepackage{mathrsfs}

\usepackage{bm}

\newtheorem{remark}{Remark}

\def\beq{\begin{equation}}
\def\eeq{\end{equation}}
\newcommand{\bea}{\begin{eqnarray}}
\newcommand{\eea}{\end{eqnarray}}
\def\bi{\begin{itemize}}
\def\ei{\end{itemize}}
\def\ba{\begin{array}}
\def\ea{\end{array}}
\def\bfig{\begin{figure}}
\def\efig{\end{figure}}

\def\lie{\mathcal{L}}
\def\scri  {\mathscr{I}}

\newtheorem{theorem}{Theorem}[section]
\newtheorem{definition}{Definition}[section]

\newtheorem{lemma}[theorem]{Lemma}
\usepackage{float}

\def\be{\begin{eqnarray}}
\def\ee{\end{eqnarray}}

\newcommand{\calr}{\mathcal R}

\newcommand{\sm}{\mathscr{M}}

\renewcommand{\o}{\omega}
\renewcommand{\O}{\Omega}

\newcommand{\lt}{\left}
\newcommand{\rt}{\right}

\newcommand{\im}{\mathrm{Im}}

\title{Null infinity as $SU(2)$ Chern–Simons theories and its quantization}

\author[1]{\ Hongwei Tan}

\author[1]{\ Kui xiao}  

\author[1]{\ Shoucheng Wang}

\affiliation[1]{School of Science, Hunan Institute of Technology, Hengyang 421002, China}

\emailAdd{honweitan@hnit.edu.cn}
\emailAdd{Corresponding author:xiaokui@hnit.edu.cn}
\emailAdd{Corresponding author:wscdaxia@qq.com}

\abstract{
This paper studies the quantization of the future null infinity ($\scri^+$) of an asymptotically flat spacetime.
Based on the observation by Ashtekar and Speziale that $\scri^+$ can be regarded as a weakly isolated horizon, we adopt the quantization framework developed for weakly horizon to quantize $\scri^+$.
We first show that the symplectic structure of $\scri^+$ is equivalent to the sum of the symplectic structures of two $SU(2)$ Chern–Simons theories with opposite levels.
Based on this observation, we apply Chern-Simons quantization approach to quantize $\scri^+$.
Finally, we compute the entropy of $\scri^+$ by counting the microstates,
showing that it is proportional to the area of $\tilde{\Delta}$, a spacelike cross-section of $\scri^+$.
Our result is consistent with the universal entropy formula in the framework of (weakly) isolated horizon.
}


\begin{document}

\maketitle

\section{Introduction}
Loop quantum gravity (LQG), a nonperturbative and background independent theory of quantum gravity, has drawn broad attention within the community in recent decades (see, e.g., \cite{han2007fundamental,thiemann2008modern,ashtekar2004background}).
When applied to cosmology, LQG replaces the classical big bang singularity with a bounce, thereby resolving the cosmological singularity problem \cite{ ashtekar2006quantum,ashtekar2011loop,bojowald2008loop,ashtekar2009loop}.
Similar techniques have also been applied to BHs, aiming to resolve the BH singularity \cite{husain2022quantum,zhang2023loop,ashtekar2005quantum,bohmer2007loop,chiou2008phenomenological,bodendorfer2021quantum}.
Furthermore, recent developments involving black-to-white hole transition scenario have offered profound insights into the BH information paradox \cite{han2024spin,han2023geometry,rignon2022black,bianchi2018white,olmedo2017black}.
In recent years, progress in loop quantum-modified black hole models has deepened our understanding of both black hole entropy and the effective dynamics of LQG \cite{lewandowski2023quantum,yang2023shadow,stashko2024quasinormal,zhang2023black,gong2024quasinormal,yang2025gravitational,liu2024gravitational,zi2025eccentric,nozari2025circular,tan2025black,tan2025massive,lin2024effective,shi2024higher}.
However, despite the significant achievements of LQG in exploring quantum aspects of gravity, the quantization of spacetimes with asymptotically flat boundary conditions has received relatively little attention within the LQG community. 
Since asymptotically flat spacetime plays significant role in gravitational system, addressing this gap may offer valuable insights into the fundamental structure of quantum gravity.

Asymptotically flat spacetimes play a crucial role in GR, particularly in the description of isolated systems. Intuitively, asymptotic flatness means that the spacetime curvature vanishes as the radial coordinate \( r \) approaches infinity. Based on this intuition, the metric of an asymptotically flat spacetime in the limit \( r \to \infty \) takes the following form (see, e.g., \cite{han2024symmetry, strominger2018lectures}):
\be
g_{\mu\nu}=\eta_{\nu\nu}+\frac{f_{\mu\nu}\lt(X^0,\vec{X}/r\rt)}{r}+O\lt(r^{-2}\rt),\label{asympflat}
\ee
with $X^\mu$ is the asymptotic Cartesian coordinate, $\eta_{\mu\nu}$ is the standard Minkowski metric, and $f_{\mu\nu}\lt(X^0,\vec{X}/r\rt)$ represents the metric on a 2-sphere.
The radial coordinate $r$ is given by $r=\sqrt{(X^1)^2+(X^2)^2+(X^3)^2}$.
The null infinity, serving as a null boundary of an asymptotically flat spacetime, plays a crucial role in modern theoretical physics (see eg., \cite{strominger2018lectures}).
In this paper, we will focus on the future null infinity, denoted by $\scri^+$, which provides a natural setting for investigating the radiative aspects of gravity and scattering problems \cite{bondi1960gravitational,ashtekar2024null,frauendiener2025fully,strominger2016gravitational,he2015bms}.
Unlike in Minkowski spacetime, the asymptotic symmetry group on $\scri^+$ is infinite dimensional and is known as the BMS (Bondi-Metzner-Sachs) group \cite{bondi1962gravitational, sachs1962gravitational}.
Recent developments of the flat space holography, where $\scri^+$ serves as the boundary of the construction, have attracted significant attention from the research community (see e.g.,\cite{donnay2023bridging}).

Penrose, Geroch and Ashtekar reformulated the definition of asymptotically flat spacetime with Penrose conformal completion construction \cite{geroch1977asymptotic,ashtekar1980null,wald2000general}.
In this approach, a physical spacetime $\left(\mathscr{M},g_{ab}\right)$ and an unphysical spacetime $\left(\tilde{\mathscr{M}},\tilde g_{ab}\right)$ are introduced,
where $\tilde{g}_{ab}=\O^2g_{ab}$ defines a conformal transformation of the physical metric, and $\O$ is known as the conformal factor.
Several requirements are introduced to $\O$ to properly define asymptotic flatness, see def. \ref{def:Asymptotic} for details.
Based on this framework, Ashtekar introduced a quantization of the gravitational radiative modes at $\scri^+$ \cite{ashtekar1987asymptotic,ashtekar2014geometry}.
Despite the significant progress in the study of $\scri^+$, a complete theory for the quantization of the spacetime background at $\scri^+$ remains elusive.
This paper aims to address this gap by quantizing $\scri^+$ using the Chern-Simons theory quantization strategy introduced in \cite{witten1989quantum}.
Our approach not only potentially deepens the understanding of the nature of $\scri^+$, but also offers potential insights into gravitational radiation and flat space holography.

Recent developments by Ashtekar and Speziale potentially lead to a strategy for quantizing $\scri^+$ \cite{ashtekar2024charge}. 
In these works, the authors demonstrate that $\scri^+$  can be identified as a weakly isolated horizon (WIH) within the so-called divergence-free conformal frame, defined by the condition $\tilde{\nabla}_a\tilde{n}^a\hat=0$, where $\hat=$ means the equation is evaluated on $\scri^+$.
Here, $\tilde{\nabla}_a$ is the covariant derivative compatible with the unphysical metric $\tilde{g}_{ab}$, and $\tilde{n}^a$ is the normal of $\scri^+$, defined via the conformal factor as $\tilde{n}_a:=\tilde{\nabla}_a\O|_{\scri^+}$.
Moreover, one can show that the imaginary part of the Weyl tensor $\tilde{\Psi}_2$ vanishes in this confomal frame.
The definition of $\tilde\Psi_2$ can be found in \eqref{eq:def_psi_2}.
This finding inspires us to apply the quantization strategy developed for WIHs to the case of $\scri^+$.

The quantization of the isolated horizons (IHs) has been investigated broadly in the literatures \cite{perez2017black,ashtekar2004isolated}.
In the spherically symmetric case, it has been shown that the classical degrees of freedom of an IH can be described by a $U(1)$ Chern-Simons theory  \cite{ashtekar1998quantum,ashtekar2000quantum}
However, as demonstrated in \cite{perez2017black}, this theory is incompatible with the Heisenberg's uncertainty principle at the quantum level.
This issue is resolved by enhancing this theory to an $SU(2)$ Chern-Simons theory \cite{engle2010black}.
Ref. \cite{perez2011static} explores the cases of static IHs, which includes horizons of arbitrary shape but without angular momentum.
The authors show that such IHs can be dynamically described by two $SU(2)$ Chern-Simons theories.
Henceforth, one can apply Chern-Simons quantization scheme introduced in \cite{witten1989quantum} to quantize this type of IHs.
Furthermore, Ref. \cite{frodden2014modelling} extends the investigations to IHs with angular momentum.

Indeed, the approach introduced in Ref. \cite{perez2011static} can be generalized to the case of WIH directly, as we will see later.
Therefore, we adopt this quantization scheme to quantize $\scri^+$, which is the main task of this paper.
We begin by briefly reviewing the results of \cite{ashtekar2024null}, which establish the equivalence between $\scri^+$ and a WIH under the divergence-free conformal frame.
This observation motivates us to apply the approach developed in  Ref. \cite{perez2011static} to the study of $\scri^+$.
In this work, we consider the vacuum action for the unphysical spacetime, i.e., with no matter fields coupled to gravity, as shown in eq. \eqref{eq:action}.
We demonstrate that this choice of action satisfies the definition of the asymptotically flat spacetime outlined in def. \ref{def:Asymptotic} by requiring $\tilde{\nabla}_a\tilde{n}^a$ falls off at least as fast as $\O^2$ While approaching to $\scri^+$, which strengthens the condition of the divergence-free conformal frame mentioned previously.
Furthermore, we show that the classical degrees of freedom can be dynamically described by a pair of $SU(2)$ Chern-Simons theories with opposite levels.
The cases with matter fields will be explored in our future works, as they may lead to nontrivial symplectic structures associated with gravitational radiation and BMS charges.

The result discussed above motivates us to quantize $\scri^+$ using Chern–Simons quantization scheme.
In this paper, we adopt the approach introduced in ref. \cite{witten1989quantum} to carry out the quantization of $\scri^+$.
Specifically, for a spacelike cross-section $\tilde{\Delta}$ of $\scri^+$, we demonstrate that the corresponding Hilbert space is given by the intertwiner space of the quantum deformation of $SU(2)$, known as $U_q(su(2))$.
This result naturally introduces a cutoff $|k|/2$ to the spin number $j$, with $k$ is the level of the Chern-Simons theory.
Finally, following the approach in \cite{ghosh2006counting}, we show that the entropy associated with $\scri^+$ is proportional to the area of $\tilde{\Delta}$.

The present work is the first step toward developing the quantization theory of $\scri^+$, where the gravitational radiation, especially the BMS charges, have not been taken into account.
To properly consider gravitational radiation and BMS charges, a more general action probably be needed for the unphysical spacetime in the discussions, which potentially provides the symplectic structure of these physical quantities.
These studies might provide us with the whole picture of the quantum theory of gravity within the context of asymptotic flatness background.

One of the most significant properties of the BMS group is the presence of an infinite number of supertranslation symmetries, which generalize the standard spacetime translations. 
These symmetries have profound physical implications: for example, it has been shown that the action of a supertranslation corresponds to the so-called gravitational memory effect, which is a detectable phenomenon observable through astronomical measurements \cite{strominger2016gravitational}.
Moreover, by considering the scattering progresses on $\scri^+$, it has been demonstrated that the supertranslation symmetry is equivalent to Weinberg’s soft graviton theorem \cite{he2015bms}. 
Hence, the investigations of the quantization theory of BMS charges may yield valuable insights into the quantum nature of gravitational waves.
Our present work provides a potential approach to study these intriguing topics, which we will continue in future research.

The Carrollian holography, as a promising approach to the flat space holography, has drawn significant attention in recent years \cite{banerjee2023carroll,bagchi2023magic,bagchi20143d,bagchi2022scattering}.
This approach is based on the observation that the  conformal Carroll group on $\scri^+$ is isomorphic to the BMS group associated with the asymptotically flat spacetime in the bulk \cite{duval2014conformal,bagchi2010correspondence,bagchi2012bms}.
Also, asymptotically flat spacetime can be regarded as a specific limit of the asymptotically anti-de Sitter (AdS) spacetime when the Cosmological constant $\Lambda$ approaches zero.
On the boundary side,  this corresponds to taking the ultra-relativistic limit $c\to0$, which induces a Carrollian limit in the boundary conformal field theory (CFT) \cite{ciambelli2018flat}.
This insight indicates that the flat holography is a specific limit of the AdS/CFT corresponding \cite{donnay2023bridging}.
Henceforth, the investigations of quantization at $\scri^+$ will benefit to better understanding of these issues.
We plan to pursue these intriguing directions in future work.

This paper is organized as following:
In Sec. \ref{sec:WIH_rev}, we provide a brief review of WIH, including its definition and key properties.
In Sec. \ref{sec: scri_as_iso}, we first present the definition of the asymptotically flat spacetime.
We then review briefly that $\scri^+$ is equivalent to a WIH under the divergence-free conformal frame.
In Sec. \ref{sec:wih_as_su2}, we show that $\scri^+$ can be dynamically described by two Chern-Simons theories.
In Sec. Sec. \ref{sec:quanti_scri}, we apply the Chern-Simons quantization scheme to $\scri^+$ and compute its entropy by counting the corresponding microstates.
Conclusions and discussions are given in Sec. \ref{sec:con_out}.

\section{Geometry of the weakly isolated horizon}\label{sec:wih}\label{sec:WIH_rev}
In this section, we provide a brief review of the definition and key properties of WIH.
WIH is a special class of non-expanding horizon (NEH), and an NEH is defined as  (e.g., see \cite{lewandowski2006symmetric,ashtekar2022non}):
\begin{definition}\label{def:NEH}
    In a 4-dimensional spacetime $(\sm, g_{ab})$, a null hypersurface $\Delta$ is an NEH if  
    \begin{enumerate}[label=(\roman*)]
        \item The topology of $\Delta$ is $\hat{\Delta}\times \mathbb{R}$, with the $\hat{\Delta}$ is an $S^2$.
        \item Given $l^a$ as the null normal of $\Delta$, its $\theta_{(l)}$ vanishes. 
        \item Einstein's equations hold on $\Delta$ and the Ricci tensor $R_{ab}$ of $g_{ab}$ satisfies $R_b^{\, a}l^b\hat{=}\alpha l^a$, with $\alpha$ is a function on $\Delta$ and $\hat{=}$ denotes that the equation is evaluated on $\Delta$.
      \end{enumerate} 
\end{definition}
Given a vector field $l^a$ as the null normal of $\Delta$, Einstein's equations combing with  requirement $(iii)$ in def. \ref{def:NEH} imply 
\begin{equation}
    8\pi GT_{ab}l^al^b
    \hat{=}R_{ab}l^al^b
    +\frac{1}{2}Rg_{ab}l^al^b
    \hat{=}\alpha l^al_a
    \hat{=}0.
\end{equation}
Hence, the physical interpretation of trequirement $(iii)$ in def. \ref{def:NEH} is that there is no matter field crossing $\Delta$.
The Raychaudhuri equation for a null hypersurface generalized associated with $l^a$ reads \cite{wald2010general}
\begin{equation}
    l^c\partial_c\theta_{l}
    \hat=-\frac{1}{2}\theta^2_{(l)}
    -\sigma_{(l)ab}\sigma^{ab}_{(l)}
    -R_{ab}l^al^b
    +\kappa_{(l)}\theta_{(l)}.
\end{equation}
According to the vanishing of $\theta_{(l)}$, we find the shear of $\Delta$ vanishes
\begin{equation}
    \sigma^{ab}_{(l)}\hat=0.
\end{equation}
The spacetime metric $g_{ab}$ induces a reduced "metric", denoted as $q_{ab}$, on $\Delta$.
However, since $\Delta$ is a null hypersurface, $q_{ab}$ is degenerated, with signature $\lt(0,\,+,\,+\rt)$.
The vanishing expansion $\theta_{(l)}$ and the shear $\sigma^{ab}_{(l)}$ imply that $\mathcal{L}_{\vec{l}}q_{ab}\hat{=}0$, which means $q_{ab}$ is invariant under the "time" evolution generated by $l^a$, see \cite{lewandowski2006symmetric}.

From the fact that both the expansion and the shear vanish on $\Delta$, one can uniquely determine a derivative operator $D_a$ on $\Delta$, which is compatible with $q_{ab}$:
\begin{equation}
    D_aq_{bc}\hat=0.
\end{equation}
With this construction, we can define a one-form field $\boldsymbol{\o}_{(l)}$  corresponding to $l^a$ as
 \begin{equation}\label{eq:def_of_ome}
    D_bl^a:\hat=\omega_{(l)b}l^a.
 \end{equation}
Since $l^a$  has the freedom of rescaling, the one-form field $\boldsymbol{\o}_{(l)}$ is not uniquely defined.
 Supposed $l'^a\hat=fl^a$, with $f$ is a smooth, positive function on $\Delta$. Then the $\boldsymbol{\o}_{(l')}$ corresponding to $l'^a$  is given by
 \begin{equation}
    \omega_{(l')b}l'^a\hat=D_bl'^a\hat=fD_bl^a+l^aD_bf
    \hat=f\omega_{(l)b}l^a+l^aD_bf.
 \end{equation}
It follows that 
\begin{equation}
\omega_{(l')a}\hat=\omega_{(l)a}+D_a\ln f.
\end{equation}
This rescaling freedom can be fixed by imposing the divergence-free condition on  $\omega_{(l)a}$
\begin{equation}
    D^a\omega_{(l)a}\hat=q^{ab}D_a\omega_{(l)b}\hat{=}0,
\end{equation}
with $q^{ab}$ is one of the inverse metrics of $q_{ab}$, which is not uniquely defined, due to the degeneracy of $q_{ab}$.
This condition fixes the scalar freedom of $l^a$ up to a constant $C$, see \cite{ashtekar2024null}. 
Therefore, we can introduce an equivalence class of $l^a$: 
$l^a$ and $l'^a$ are considered equivalent if $l'a=Cl^a$, with $C$ is a constant. The equivalence class is denoted by $[l^a]$.

A WIH is a type of special NEH, equipped with an equivalence class $[l^a]$, satisfying (e.g., see \cite{ashtekar1998quantum, ashtekar2000generic})
\begin{equation}\label{eq:def_WIH}
\lie_{\vec{l}}D_bl^a\hat{=}l^a\lie_{\vec{l}}\omega_{(l)b}\hat=0.
\end{equation}
Here we have used $\lie_{\vec{l}}l^a\hat=0$.
Equivalently, eq. \eqref{eq:def_WIH} can be rewritten as
\begin{equation}\label{eq:def_WIH_2}
    0\hat{=}\lie_{\vec{l}}\omega_{(l)a}
    \hat{=}\lt[\lie_{\vec{l}}D_a,D_a\lie_{\vec{l}}\rt]l^b.
\end{equation}
Consequently, we can now introduce the definition of the WIH:
\begin{definition}\label{def:wih}
    A WIH $\lt(\Delta,\lt[l^a\rt]\rt)$ is an NEH equipped with an equivalence class $\lt[l^a\rt]$, such that $\forall\, l^a \in \lt[l^a\rt]$, the following condition holds:
    \begin{equation}\label{eq:def_WIH_fin}
        \lt[\lie_{\vec{l}}D_a,D_a\lie_{\vec{l}}\rt]l^b\hat{=}0.
    \end{equation}
\end{definition}
The surface gravity $\kappa_{(l)}$ of the WIH is given by 
\begin{equation}
    \kappa_{(l)}:=l^a\omega_{(l)a}.
\end{equation}
It has been proven that $\kappa_{(l)}$ is constant on the WIH \cite{ashtekar2000generic}. 
In particular, if $\kappa_{(l)}=0$,  $\lt(\Delta,\lt[l^a\rt]\rt)$ is an extremal WIH.
It is useful to introduce the Newman-Penrose (NP) frame $\lt(l^a,n^a,m^a,\bar{m}^a\rt)$ for further discussion \cite{newman1962approach,chandrasekhar1985mathematical}.
Where $l^a$ is chosen from the equivalence class $\lt[l^a\rt]$.
$n^a$ is given by $n_a=d(v)_a$, with $v$ is the affine parameter along $l^a$. As a result, the normalization condition $g_{ab}l^an^a=-1$ holds.
The null hypersurface $\Delta$ is foliated by $v=constant$ 2-sphere with $l^a$ and $n^a$ as its null normals.
$m^a$ is a complex null vector, tangents to these 2-spheres, satisfies $\lie_{\vec{l}}m^a=0$,
and its complex conjugate $\bar m^a$ obeys normalization condition $g_{ab}m^a\bar{m}^b=1$.
The volume element of each 2-sphere cross-section is then given by 
\begin{equation}
    \epsilon_{ab}\equiv im_a\wedge\bar{m}_b.
\end{equation}
The components of the spacetime metric in this frame are then expressed as
\begin{equation}
    \left(g_{\mu \nu}\right)=\left[\begin{array}{cccc}
    0 & 1 & 0 & 0 \\
    1 & 0 & 0 & 0 \\
    0 & 0 & 0 & -1 \\
    0 & 0 & -1 & 0
    \end{array}\right].
    \end{equation}
The following components of the Weyl tensor in the NP frame will be useful:
\bea 
\Psi_0&:=&C_{abcd}l^am^bl^cm^d\hat{=}0,\\
\Psi_1&:=&C_{abcd}l^am^bl^cn^d\hat{=}0,\\
\Psi_2&:=&C_{abcd}l^am^b\bar{m}^cn^d\hat{=}\text{Re}\,\Psi_2+i\text{Im}\,\Psi_2,\label{eq:def_of_psi_2}\label{eq:def_psi_2}
\eea
with
\begin{equation}
    \text{Re}\,\Psi_2\hat{=}-\frac{1}{4}\calr+\frac{1}{24}R,\quad
    \text{Im}\,\Psi_2\hat{=}D_{[a}\omega_{(l)b]}.
\end{equation}
Here, $R$ denotes the Ricci scalar of the spacetime metric, while $\calr$ represents the Ricci scalar of the $S^2$.
$\text{Im}\,\Psi_2$ defines the angular momentum of the WIH \cite{ashtekar2004multipole}.
Based on the value of $\text{Im}\,\Psi_2$, the WIH is classified into the following two categories:
\begin{itemize}
    \item \textbf{Static}: In these cases,
            \begin{equation}\label{eq:reiu_static}
                \text{Im}\,\Psi_2=0.
            \end{equation}
        This means that the horizon is not rotating. 
    \item \textbf{Non-Static}: In these cases,
    \begin{equation}
        \text{Im}\,\Psi_2\neq0.
    \end{equation}
\end{itemize}
\section{$\scri^+$ as a weakly isolated horizon}\label{sec: scri_as_iso}
In this section, we briefly review of the work by Ashtekar and Speziale, which demonstrates that   $\scri^+$ is equivalent to a WIH in a divergence-free conformal frame \cite{ashtekar2024null, ashtekar2024charge}.

The definition of an asymptotically flat spacetime based on Penrose's conformal completion strategy is given as follows \cite{geroch1977asymptotic,ashtekar1980null,wald2000general}: 
\begin{definition}
    A Physical spacetime $\left(\mathscr{M},g_{ab}\right)$ is asymptotically flat if there exists a manifold $\tilde{\mathscr{M}}$ with $\exists\, i^0\in\tilde{\mathscr{M}}$ such that $\overline{\mathrm{J}^{+}\left(i^0\right)} \cup \overline{\mathrm{J}^{-}\left(i^0\right)}=\tilde{\mathscr{M}}-\mathscr{M}$.
    With $i^0$ is the spatial infinity. The boundary of $\tilde{\mathscr{M}}$ is $\dot{\mathscr{M}}=\mathscr{I}^{+} \cup \mathscr{I}^{-} \cup\left\{i^0\right\}$. 
    Here, $\mathscr{I}^+$ ($\mathscr{I}^-$) is called the future (past) null infinity.
    The metric $\tilde{g}_{ab}$ on $\tilde{\mathscr{M}}$ is a conformal transformation of $g_{ab}$, $\tilde{g}_{ab}=\O^2g_{ab}$,
    with $\O$ is a non-vanishing function on $M$.
    $\left(\mathscr{M},g_{ab}\right)$ and $\left(\mathscr{
        \tilde M},\tilde g_{ab}\right)$ satisfy:
    \begin{enumerate}[label=(\alph*)]\label{def:Asymptotic}
        \item $\scri^{+}$ has topology $S^2\times R$ and its causal past contains a non-empty portion of $\mathscr{M}$.
        \item $\O$ is smooth on $\tilde{M}-\{i^0\}$ and $C^2$ at $i^0$.
        \item $\left.\Omega\right|_{\dot{\mathscr M}}=0$, $\left.\tilde{\nabla}_a \Omega\right|_{\mathcal{J}^{ \pm}} \neq 0$ and $\lim _{\rightarrow i^0} \tilde{\nabla}_a \Omega=0$, with $\tilde{\nabla}_a$ is compatible with $\tilde{g}_{ab}$.
        \item The physical metric satisfies  Einstein’s equations $R_{a b}-\frac{1}{2} R g_{a b}=8 \pi G T_{a b}$, where $\O^{-1}T_{ab}$ admits a smooth limit to $\mathscr{I}^+$.
    \end{enumerate}
\end{definition}
In this definition, $\left(\tilde{\mathscr{M}},\tilde{g}_{ab}\right)$ is called the unphysical spacetime, which corresponds to the conformal completion of the  physical spacetime $\left(M,g_{ab}\right)$.
The vector field $\tilde{n}_a:=\tilde{\nabla}_a\O$ is the normal of $\scri^+$.
Einstein's equations can be rewritten in terms of the conformally rescaled  quantities
\begin{equation}\label{eq:confomal_Einstein}
    \tilde{R}_{a b}-\frac{1}{2} \tilde{R} \tilde{g}_{a b}+2 \Omega^{-1}\left(\tilde{\nabla}_a \tilde{n}_b-\left(\tilde{\nabla}^c\tilde{n}_c\right) \tilde{g}_{a b}\right)+3 \Omega^{-2}\left(\tilde{n}^c \tilde{n}_c\right) \tilde{g}_{a b}=8 \pi G T_{a b},
    \end{equation}
    where $\tilde R_{ab}$ is the Ricci tensor of $\tilde g_{ab}$.
    In the following discussions, we focus on the physics at $\mathscr{I}^+$.
    Eq. \eqref{eq:confomal_Einstein}, together with assumption (d) on the fall-off condition of the stress-energy tensor implies the following properties of $\scri^+$ \cite{geroch1977asymptotic,ashtekar2014geometry}:
    \begin{enumerate}
        \item $\tilde{n}^a\tilde{n}_a\hat{=}0$. Here $\hat{=}$ means the equation is evaluated at $\mathscr{I}^+$.
        Since $\tilde{n}^a$ is the normal of $\mathscr{I}^+$, this property indicates $\mathscr{I}^+$ is a null hypersurface.
        \item There is a degree of confomal freedom for  the choice of the conformal factor $\O$.
            Given  $\O$ satisfying the requirements in def. \ref{def:Asymptotic}, a rescaled function $\O'=\mu\O$, with $\mu$ is nowhere vanishing on $\tilde{\mathscr{M}}$, also defines a conformal completion.
            This freedom can be fixed by imposing the divergence-free condition on the conformal frame, i.e., $\tilde{\nabla}_a\tilde{n}^a\hat{=}0$.
           We call this conformal frame as the divergence-free conformal frame.
            With eq. \eqref{eq:confomal_Einstein}, the divergence-free conformal frame indicates 
            \begin{equation}
                \tilde{\nabla}_a\tilde{n}^b\hat{=}0.
            \end{equation}
            \item The Schouten tensor of $\tilde{g}_{ab}$, defined as $\tilde{S}_a{ }^b=\tilde{R}_a{ }^b-\frac{1}{6}, \tilde{R} \delta_a{ }^b$,  satisfies $\tilde{S}_a{ }^b \tilde{n}^a \hat{=}-\tilde{f} \tilde{n}^b$, where $\tilde{f}=\O^{-2}\tilde{n}^a\tilde{n}_a$.
            Therefore, condition (c) in def. \ref{def:Asymptotic} implies $\tilde{R}_a{ }^b \tilde{n}^a \hat{=} \alpha \tilde{n}^b$, with $\alpha \hat{=} \frac{1}{6} \tilde{R}-\tilde{f}$.
            \item  The Weyl tensor $\tilde{C}_{abc}{}^d$ vanishes on $\mathscr{I}^+$.
            Hence if $\tilde{g}_{ab}$ is $C^3$ at $\mathscr{I}^+$, then $\tilde{K}_{abcd}:=\O^{-1}\tilde{C}_{abcd}$ admits a continuous limit to $\mathscr{I}^+$.
    \end{enumerate}
Properties 1-3 indicate that $\mathscr{I}^+$ is a NEH in the unphysical spacetime $\left(\tilde{\mathscr{M}},\tilde{g}_{ab}\right)$.
Indeed, $\tilde{n}^a$ here plays the same role as $l^a$ introduced in section \ref{sec:wih}.
Compared with eq. \eqref{eq:def_of_ome}, the divergence-free condition $\tilde{\nabla}_a\tilde{n}^a\hat{=}0$ indicates that the associated one-form field $\boldsymbol{\o}_{(\tilde{n})}$  vanishes.
Therefore, we have
\begin{equation}
    \mathcal{L}_{\vec{\tilde{n}}}\boldsymbol{\o}_{(\tilde{n})}=0.
\end{equation}
Compared with eq. \eqref{eq:def_WIH_2}, we find that satisfies that $\mathscr{I}^+$ satisfies the conclusions of a WIH. 
Furthermore, the vanishing of $\boldsymbol{\o}_{(\tilde{n})}=0$ indicates that the the surface gravity of $\scri^+$ also vanishes.
Hence, it is an extremal WIH. 
Similar to the case of a WIH, the Weyl scalar $\tilde{\Psi}$ of $\scri^+$ can be decomposed into its real part and imaginary parts:
\begin{equation}
    \tilde{\Psi}_2
    =\text{Re}\tilde{\Psi}_2
    +\text{Im}\tilde{\Psi}_2,
\end{equation} 
with
\begin{equation}
    \text{Im}\tilde{\Psi}_2
    =\tilde{D}_{[a}\o_{(\tilde{n})b]}.
\end{equation}
Here,  $\tilde{D}$ is the covariant derivate compatible to $\tilde{q}_{ab}$, with $\tilde{q}_{ab}$ is the reduced metric of $\tilde{g}_{ab}$ on $\scri^+$.
Furthermore, $\boldsymbol{\o}_{(\tilde{n})}=0$ indicates
\begin{equation}
    \text{Im}\tilde{\Psi}_2=0.
\end{equation}
\section{$\scri^+$ as $SU(2)$ Chern-Simon theories}\label{sec:wih_as_su2}

In Ref. \cite{perez2011static}, it is shown that a static WIH can be described dynamically by two Chern-Simons theories.
In this section, we will see that this result also holds on $\scri^+$, based on the fact that $\scri^+$ is a WIH with $\text{Im}\tilde{\Psi}_2=0$.
\subsection{The main equations on $\scri^+$}
In this work, the topology of  the unphysical spacetime manifold $\tilde{\mathscr{M}}$ is assumed to be $\tilde{M}\times R$, where $\tilde M$ is a spacelike hypersurface representing the spatial slice of the foliation.
The cross-section $\tilde{\Delta}=\tilde{M}\cap \scri^+$ is a spacelike 2-sphere.
For further discussions, it is convenient to reformulate  the theory using the Ashtekar-Barbero variables, as introduced in \cite{ashtekar1986new,ashtekar2004background,han2007fundamental,rovelli2004quantum}.  
Since $\scri^+$ lies within the unphysical spacetime $\left(\mathscr{\tilde M},\tilde g_{ab}\right)$,
we express the unphysical metric in terms of the co-tetrad as
\begin{equation}\label{eq:def_tet}
    \tilde g_{ab}=e_a^{\,I}e_b^{\,J}\eta_{IJ}.
\end{equation}
The internal index $I$ refers to the Minkowski index $I=(0,\,i)$, where $i$ takes the value of $i=1,\,2,\,3$.
The internal $\eta_{IJ}$  is the Minkowski metric, given by $\eta_{IJ}=(-1,\,1,\,1,\,1)$.
The tetrad $e^a_I$ and the co-tetrad $e_i^{\,J}$ are inverse to each other, satisfying
\begin{equation}
    e_a^{\,I}e_{I}^{\,b}=\delta_a^{\,b},\quad
    e_a^{\,J}e_{I}^{\,a}=\delta_I^{\,J}.
\end{equation}
With the tetrads introduced previously, we define a connection one-form $\o_a^{IJ}$ as
\begin{equation}\label{eq:def_of_connection}
    \omega_a^{\,IJ} :=e^{Ib}\tilde\nabla_ae_b^{\,J}.
\end{equation}
It is obvious that the internal indices of $\omega_a^{\,IJ}$ are antisymmetric: $\omega_a^{IJ}=-\omega_a^{JI}$.
Equivalently, $\o_a^{IJ}$ satisfies the following relation:
\begin{equation}\label{eq:eq_for_ome}
    \tilde{\nabla}_be^I_a=-\o^I_{bJ}e_a^J.
\end{equation}

A family of hypersurfaces are introduced by 3+1 decomposition of the spacetime, whose extrinsic curvatures are denoted as $K_{ab}$.
With the tetrads introduced previously, we define the extrinsic curvature 1-form as
\begin{equation}
    K_a^{\,i}:=K_{a}^{\,b}e_b^{\,i},
\end{equation}
which satisfies the relation given in \cite{thiemann2008modern}
\begin{equation}\label{eq:rela_K_ome}
    K_a^{\,i}=2\omega_a^{0i}.
\end{equation}
Then the Ashtekar self-dual connection is defined as 
\begin{equation}
    A_a^{+i}:=\Gamma_a^{\,i}+i K_a^{\,i},
\end{equation}
where $\Gamma_a^{\,i}$ is the spin connection, given by \cite{thiemann2008modern}
\begin{equation}
    \Gamma_a^{\,i}=\frac{1}{2}\epsilon^{ijk}e_k^{\,b}
    \left(
        \partial_be_a^{\,j}
        -\partial_ae_b^{\,j}
        +e_j^{\,c}e_a^{\,l}\partial_be_c^{\,l}
    \right).
\end{equation} 
The curvature 2-form of $ A_a^{+i}$ is given by 
\begin{equation}\label{eq:curva_forsel_dual}
    F^i(A^+):=dA^{+i}+\frac{1}{2}\epsilon^i_{\,jk}A^{+j}\wedge A^{+k},
\end{equation}
with $\epsilon_{ijk}$ being the Levi-Civita symbol.
The spatial tetrads $\{e_i^{\,a}\}$ possess an $SU(2)$ gauge freedom. 
Before proceeding further, we introduce the following gauge fixing:
$e_1^{\,a}$ is normal to $\hat{\Delta}$, and both $e_2^{\,a}$ and $e_3^{\,a}$ are tangent to $\hat{\Delta}$.
Then the NP frame on $\scri^+$ is given by
\begin{equation}\label{eq:gauge_fixing}
    \tilde{n}^a=\frac{1}{2}\lt(e_0^{\,a}+e_1^{\,a}\rt),\quad
    \tilde{l}^a=\frac{1}{2}\lt(e_0^{\,a}-e_1^{\,a}\rt),\quad
    \tilde{m}^a=\frac{1}{2}\lt(e_2^{\,a}+ie_3^{\,a}\rt).
\end{equation}
The bi-vector 2-form is defined as
\begin{equation}\label{eq:def_of_bi_vector}
    \Sigma^i=\epsilon^i_{\,jk}e^j\wedge e^k,
\end{equation}
Cartan's second structural equation implies 
\begin{equation}
    F_{ab}^i(A^+)=-\frac{1}{4}\tilde{R}_{ab}^{\,\,\,\,cd}\Sigma_{ab}^{+i},
\end{equation}
where
\begin{equation} \label{eq:def_of_sig_+}
    \Sigma^{+i}=\Sigma^i+2i  e^0 \wedge e^i.
    \end{equation}
Here, $\tilde{R}_{abc}^{\,\,\,\,\,\,\,d}$ is the Reimann tensor of $\tilde{g}_{ab}$.
It is calculated in \cite{chandrasekhar1998mathematical}, which yields
\begin{equation}\label{eq:curvature_on_wih}
    \underset{\Leftarrow}{F_{ab}}^i(A^+)
    =\lt(\tilde\Psi_2-\tilde\Phi_{11}-\frac{\tilde R}{24}\rt)\underset{\Leftarrow}{\Sigma^i_{ab}}.
\end{equation}
Here,    "$\Leftarrow$" means pulling back to the $\hat{\Delta}$.
The scalar field $ \tilde\Phi_{11}$ is defined as
\begin{equation}
    \tilde\Phi_{11}:=\tilde R_{a b}\left(\tilde l^a\tilde n^b+\tilde m^a\tilde{ \bar{m}}^b\right) / 4.
    \end{equation}

The following theorems are useful for our further discussions:
\begin{theorem}\label{thm:for_0}
    Under the gauge fixing introduced earlier, we have
    \begin{equation}
\underset{\Leftarrow}{\boldsymbol{\omega}}^{21}=\underset{\Leftarrow}{\boldsymbol{\omega}}^{20},\quad
\underset{\Leftarrow}{\boldsymbol{\omega}}^{31}=\underset{\Leftarrow}{\boldsymbol{\omega}}^{30}.
\end{equation}
\end{theorem}
The proof of this theorem can be found in appendix. \ref{app:detail}.
The definition of $\boldsymbol{\o}_{(\tilde{n})}$ implies
\begin{equation}\label{eq:det_o}
    \o_{(\tilde{n})a}=-\tilde{l}_b\tilde{D}_a\tilde{n}^b,
\end{equation}
From eqs. \eqref{eq:eq_for_ome} and \eqref{eq:gauge_fixing}, we obtain
\begin{equation}
    \begin{aligned}
        \tilde{\nabla}_a\tilde{n}_b
        \hat{=}\frac{1}{\sqrt{2}}
        \tilde{\nabla}_a
        (
            e^1_{b}
            -e^0_{b}
        )
        =-\frac{1}{\sqrt{2}}
        (
            \underset{\Leftarrow}{\omega_a}^{1J}e_{Jb}
            -\underset{\Leftarrow}{\omega_a}^{0J}e_{Jb}
        ).
    \end{aligned}
\end{equation}
Then \eqref{eq:det_o} reads
\begin{equation}
    \begin{aligned}
        \o_{(\tilde{n})a}
    =-\tilde{l}^b\tilde{D}_a\tilde{n}_b
    =&\frac{1}{2}
    (
        e_1^b-e_0^b
    )
    (
            \underset{\Leftarrow}{\omega_a}^{1J}e_{Jb}
            -\underset{\Leftarrow}{\omega_a}^{0J}e_{Jb}
    )\\
    =&-\underset{\Leftarrow}{\omega_a}^{01}
    =-\frac{1}{2}\underset{\Leftarrow}{K_a}^1.
    \end{aligned}
\end{equation}
Since $\boldsymbol{\omega}_{(\tilde{n})}$ vanishes in the divergence-free conformal frame, we conclude that
\begin{equation}\label{eq:k_1_0}
    \underset{\Leftarrow}{K_a}^1=0,
\end{equation}
which leads to the following theorem:
\begin{theorem}\label{thm:thre_for_KK}
 On $\scri^+$, we have
\begin{equation}
    \underset{\Leftarrow}{K}^j \wedge \underset{\Leftarrow}{K}^k \epsilon_{i j k}=c \underset{\Leftarrow}{\Sigma}^i,
    \end{equation}
    where $c$ is a scalar  field on $\tilde\Delta$.
\end{theorem}
The proof of this theorem can be found in Appendix \ref{app:detail}.
\begin{remark}
    An isolated horizon (IH) is a special case of WIH, defined by strengthening the boundary condition \eqref{eq:def_WIH_fin} to
    \begin{equation}
        \lt[\lie_{\vec{l}}D_a,D_a\lie_{\vec{l}}\rt]v^b\hat{=}0,\quad
        \forall v\in T(\Delta).
    \end{equation}
    With this enhanced boundary condition, one can demonstrate that $\mathcal{L}_{\vec{l}}c\hat{=}0$ in the case of IH, see \cite{perez2011static} for details.
    However, since $\scri^+$ only satisfies the weaker WIH boundary condition, such nice property do not necessarily hold.
    Nevertheless, this limitation does not affect the discussions in this paper.
\end{remark}
Furthermore, $K^1_a=0$ leads to
\begin{equation}\label{eq:proper_for_static}
    v\lrcorner  \underset{\Leftarrow\Leftarrow}{K_i\Sigma^i}=0,\footnote{Here $\lrcorner$ denotes the interior product: $(\xi\lrcorner\boldsymbol{\omega})_{b_1...b_{p-1}}=\xi^a,\boldsymbol{\omega}_{ab_1...b_{p-1}}$, where $\boldsymbol{\omega}$ is a $p$-form field. It satisfies the Leibniz rule: Given $\boldsymbol{\alpha}$ a p-form field and $\boldsymbol{\beta}$ a q-form field, then $\xi\lrcorner(\boldsymbol{\alpha}\wedge\boldsymbol{\beta})=\xi\lrcorner(\boldsymbol{\alpha})\wedge\boldsymbol{\beta}+(-1)^p\boldsymbol{\alpha}\wedge\xi\lrcorner\boldsymbol{\beta}.$}.
\end{equation}
Here, $v$ is a vector field tangents to $\scri^+$. 
This result follows from the fact that $\underset{\Leftarrow}{\Sigma^A_{ab}}=\pm\epsilon^A_{\,\,\,1B}\underset{\Leftarrow}{e}^1\wedge\underset{\Leftarrow}{e}^B=0$.
\subsection{The action principle and the symplectic structure of the system}
In this work, we consider the action of the unphysical spacetime $\left(\tilde{\mathscr{M}},\tilde{g}_{ab}\right)$, given as \cite{ashtekar1999isolated,engle2010black,perez2011static}
\begin{equation}\label{eq:action}
    S\left[e, A^{+}\right]=-\frac{i}{\kappa} \int_{\tilde{\mathscr{M}}} \Sigma_i^{+}(e) \wedge F^i\left(A^{+}\right),
    \end{equation}     
    where $\kappa=16\pi G$.             
The action of the matter field is absent in eq. \eqref{eq:action}, indicating a "vacuum" of the unphysical spacetime.
This implies that the unphysical metric satisfies the vacuum Einstein's equations
\begin{equation}\label{eq:ein_eq_for_con}
    \tilde{R}_{ab}-\frac{1}{2}\tilde{R}\tilde{g}_{ab}=0.
\end{equation}
As mentioned earlier, the Einstein's equations of the physical metric can be rewritten in terms of conformal quantities as 
\begin{equation}\label{eq:con_eiin_eq}
    \tilde{R}_{ab}
    -\frac{1}{2}\tilde{R}\tilde{g}_{ab}
    +2\O^{-1}\left(\tilde{\nabla}_a\tilde{n}_b-\left(\tilde{\nabla}^c\tilde{n}_c\right)\tilde{g}_{ab}\right)+3\O^{-2}\left(\tilde{n}^c\tilde{n}_c\right)\tilde{g}_{ab}
    =8\pi G T_{ab}.
\end{equation}
Taking limit toward $\scri^+$, we find
\begin{equation}
    \lim _{\rightarrow\scri^+}\O^{-2}\left(\tilde{n}^c\tilde{n}_c\right)\tilde{g}_{ab}
    =\lim_{\rightarrow\scri^+}\left(\tilde{n}^c\tilde{n}_c\right)g_{ab}
    =0,
\end{equation} 
as $\tilde{n}^a$ is the null normal of $\scri^+$.
Then combining eq. \eqref{eq:ein_eq_for_con} with eq. \eqref{eq:con_eiin_eq}, we obtain
\begin{equation}
    \lim_{\rightarrow\scri^+}T_{ab}
    =\frac{1}{4\pi G}\lim_{\rightarrow\scri^+}\O^{-1}\left(\tilde{\nabla}_a\tilde{n}_b-\left(\tilde{\nabla}^c\tilde{n}_c\right)\tilde{g}_{ab}\right).
\end{equation}
Therefore, to fulfill the requirement $d$ of def. \ref{def:Asymptotic}, we must demand that $\tilde{\nabla}_a\tilde{n}_b$ falls off at least as fast as $\O^2$ as it approaches to $\scri^+$.
 This strengthens the earlier requirement $\tilde{\nabla}_a\tilde{n}^a|_{\scri^+}=0$.

 Later, we will show that no boundary term is needed on $\scri^+$ in the variational principle on $\scri^+$ if the allowed variations are chosen such that they preserve the boundary conditions on $\scri^+$ up to diffeomorphisms and gauge transformations.
Therefore, no  boundary term is required in eq. \eqref{eq:action} for the differentiability of the action $S\left[e, A^{+}\right]$.
The self dual version of Einstein’s equations follow from the variation of the action
    \begin{equation}\label{eq:delta_S}
       \delta S\left[e, A_{+}\right]=\frac{-i}{\kappa} \int_{\tilde{\mathscr{M}}} \delta \Sigma_i^{+}(e) \wedge F^i\left(A_{+}\right)-d_{A_{+}} \Sigma_i^{+} \wedge \delta A_{+}^i+d\left(\Sigma_i^{+} \wedge \delta A_{+}^i\right).
   \end{equation}   
They are \footnote{Here $d_{A_{+}}$ denotes the exterior derivative corresponding to the connection $A_+$. For example, $d_{A_{+}} \Sigma_i^{+}=d \Sigma_i^{+}+\epsilon_{ijk}A^j_{+}\wedge\Sigma_k^{+}$.}
\bea \label{eq: eq_of_mo}
&&\epsilon_{i j k} e^j \wedge F^i\left(A_{+}\right)+i e^0 \wedge F_k\left(A_{+}\right)=0,\nonumber\\
&&e_i \wedge F^i\left(A_{+}\right)=0,\nonumber\\
&&d_{A_{+}} \Sigma_i^{+}=0.
\eea

We denote the phase space at $\scri^+$ as $\Gamma$.
Each point $p\in \Gamma$ is parametrized by a pair $p=\lt(\Sigma^+,A_+\rt)$, where $\Sigma^+$ and $A_+$ are the dynamical variables satisfying the boundary conditions at $\scri^+$.
The tangent space of a point $p\in \Gamma$, denoted as $T_p(\Gamma)$, consists of variations $\delta =\left(\delta \Sigma^+,\delta A^+ \right)$ which describe the infinitesimal changes in the bi-vector field and the Ashtekar self-dual connection, respectively. 
As discussed earlier, the allowed variations on $\scri^+$ are restricted to those generated by diffeomorphisms and internal gauge transformations. Their actions on the fields are given by
\beq \label{eq:var_fix_WIH}
\delta e^i=\delta_\alpha e^i+\delta_v e^i,\quad
\delta\Sigma^{+}=\delta_\alpha \Sigma^{+}+\delta_v\Sigma^{+},\quad
\delta A_+=\delta_\alpha A_++\delta_v A_+,
\eeq
where $\alpha$ is the parameter of the gauge transformation, and $v$ is a vector field on $\scri^+$ generating diffeomorphisms.
Here we require $\alpha|_{\partial\scri^+}=0$, $v|_{\partial\scri^+}=0$.
The infinitesimal gauge transformations act as
\beq \label{eq:gaue_transform}
\delta_\alpha e^i=\epsilon^i_{\,jk}\alpha^je^k,\quad
\delta_\alpha\Sigma^{+i}=\epsilon^i_{\,jk}\alpha^j\Sigma^{+k},\quad
\delta_\alpha A_+=-d_{A_+}\alpha,
\eeq
The diffeomorphism-induced variations are given by
\begin{equation}\label{eq:diffeomorphism}
    \begin{aligned}
        \delta_v e^i=&v\lrcorner de^i+d\lt(v\lrcorner e^i\rt),\\
        \delta_v\Sigma^{+i}=&d_{A_+}\lt(v\lrcorner\Sigma^+\rt)^i-\epsilon_{ijk}v\lrcorner A^{+j}\Sigma^{+k},\\
        \delta_vA^{+i}=&v\lrcorner F^i\left(A^+\right)+d_{A^+}\left(v\lrcorner A^+\right).
    \end{aligned}
    \end{equation}
Here we have imposed the Gauss constraint. In \eqref{eq:delta_S}, the boundary term at $\scri^+$, arising from variational principle is given by
\begin{equation}
    B(\delta)=-\frac{i}{\kappa} \int_{\scri^+} \Sigma^+_i \wedge \delta A^i_+.
    \end{equation}
Since the allowed variations are restricted to gauge transformations and diffeomorphisms, 
this boundary term can be decomposed as $B(\delta)=B(\delta_\alpha)+B(\delta_v)$.
Based on the field equations \eqref{eq: eq_of_mo} and the transformation rules \eqref{eq:gaue_transform} and \eqref{eq:diffeomorphism}, 
each part can be shown to vanish respectively.\newline
For the gauge part:
\begin{eqnarray}
    B\left(\delta_\alpha\right)
     &=&\frac{i}{\kappa} \int_{\scri^+} \Sigma^+_i \wedge d_{A_+} \alpha^i
    =-\frac{i}{\kappa} \int_{\scri^+}\left(d_{A_+} \Sigma^+_i\right) \alpha^i
    +\frac{i}{\kappa}\int_{\partial \scri^+} \Sigma_i \alpha^i
    =0
\end{eqnarray}
For the diffeomorphism part:
\begin{eqnarray}
         B\left(\delta_v\right)&=&-\frac{i}{\kappa} \int_{\scri^+} \Sigma^+_i \wedge\left(v\lrcorner F^i(A_+)+d_{A_+}
        \left(v\lrcorner A^i_+\right)\right)\nonumber \\
        &=&-\frac{i}{\kappa} \int_{\scri^+} \Sigma^+_i \wedge\left(v\lrcorner F^i\left(A_+\right)\right)
        +\frac{i}{\kappa} \int_{\scri^+} d_{A_+} \Sigma^+_i\left(v\lrcorner A^i_+\right)
        -\frac{i}{\kappa}\int_{\partial \scri^+} \Sigma^+_i\left(v\lrcorner A^i_+\right)\nonumber\\
        &=&0.
\end{eqnarray}
        Here we have used the fact that both $\alpha$ and $v$ vanish at $\partial \scri^+$.
        Similar discussions can be applied to $\scri^-$.
        Hence, the differentiability of eq. \eqref{eq:action} is preserved without the need for additional boundary terms.

To analyze the physics on $\mathscr{I}^+$, we introduce a family of time-slices that intersect with it, as illustrated in fig. \ref{fig:slice_for_asym_flat}.
These time-slices enable the application of integration by part techniques to construct the symplectic structure on $\tilde{\Delta}$, which we will discuss in detail shortly.
\begin{figure}[H] 
	\centering 
	\includegraphics[width=0.4\textwidth]{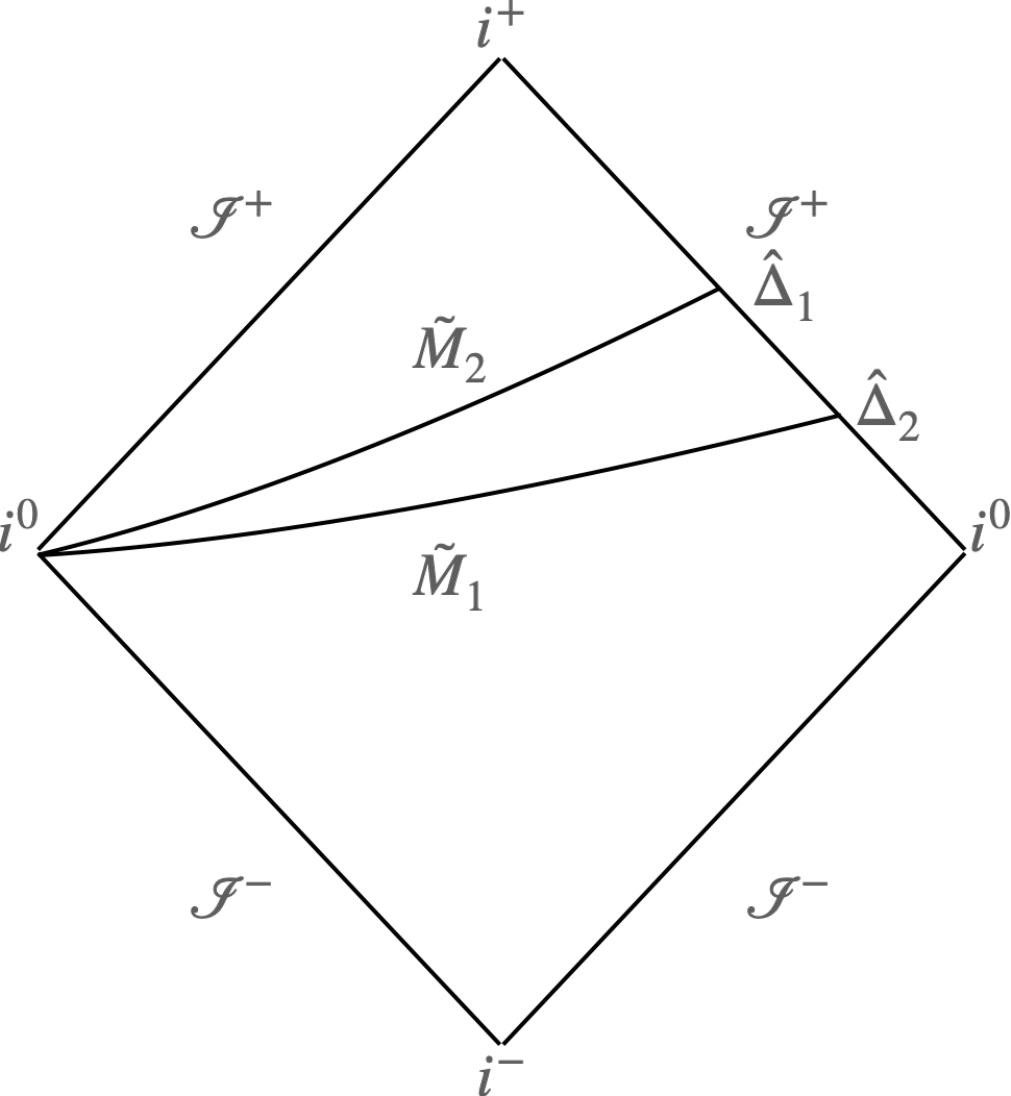} 
	\caption{The time-slices $\tilde{M}_1$ and $\tilde{M}_2$ intersect $\mathscr{I}^+$ with $\tilde{\Delta}_1$ and $\tilde{\Delta}_2$.} 
	\label{fig:slice_for_asym_flat} 
	\end{figure}
Following the covariant phase space formalism \cite{lee1990local,wald2000general}, the symplectic potential of the system is extracted from the boundary terms in \eqref{eq:delta_S}, and is given by
\begin{equation}
   \theta(\delta)=\frac{-i}{\kappa} \Sigma_i^{+} \wedge \delta A^{+i} \quad \forall \delta \in T_p \Gamma.
    \end{equation}
It is important to clarify that the "boundary" here does not refer to the boundary of the entire spacetime manifold $\tilde{\mathscr{M}}$.
As illustrated in fig. \ref{fig:slice_for_asym_flat}, it refers instead to a bounded region enclosed by $\tilde{M}_1$ and $\tilde{M}_2$, along with the portion of $\scri^+$ lying between $\tilde{\Delta}_1$ and $\tilde{\Delta}_2$, which we denote as $\hat{\scri}^+$.
The corresponding symplectic current is then given by
\begin{equation}
    J\left(\delta_1, \delta_2\right)=-\frac{2 i}{\kappa} \delta_{[1} \Sigma_i^{+} \wedge \delta_{2]} A^{+i} \quad \forall \delta_1, \delta_2 \in T_p \Gamma.
    \end{equation}
The equations of motion imply $dJ=0$. 
This is a universal result that holds independently of the specific choice of theory.
A detailed derivation can is provided in appendix \ref{app:clo_J}.
Then, fig. \ref{fig:slice_for_asym_flat} and Stokes' theorem imply
\begin{equation}\label{eq:curr_stok}
    -\int_{\tilde M_1} \delta_{[1} \Sigma_i \wedge \delta_{2]} A_{+}^i+\int_{\tilde M_2} \delta_{[1} \Sigma_i \wedge \delta_{2]} A_{+}^i+\int_{\hat\scri^+} \delta_{[1} \Sigma_i^{+} \wedge \delta_{2]} A_{+}^i=0.
    \end{equation}
According to the definition of $\Sigma^+_i$ in eq. \eqref{eq:def_of_sig_+}, when it is pulled back to the time-slice $\tilde{M}$, only the spatial part $\Sigma_i$ remaines.
Eq. \eqref{eq:curr_stok} implies there locally exists a function $\mu$ on the phase space, such that
\beq \label{eq:varia_mu}
\begin{aligned}
    \delta \mu
    =&-i\int_{\tilde{M}_1} \Sigma_i \wedge \delta A^i_+
    +i\int_{\tilde{M}_2}  \Sigma_i \wedge \delta A^i_+
    +i\int_{\hat\scri^+} \Sigma_i^{+} \wedge \delta A_{+}^i\\
    =&-i\int_{\tilde{M}_1} \Sigma_i \wedge \delta (\Gamma^i+iK^i)
    +i\int_{\tilde{M}_2}  \Sigma_i \wedge \delta (\Gamma^i+iK^i)
    +i\int_{\hat\scri^+} \Sigma_i^{+} \wedge \delta A_{+}^i.
\end{aligned}
\eeq
Then we arrive at the following theorem:
\begin{theorem}
    The symplectic structure on a time-slice $\tilde{M}$ is
\begin{equation}\label{eq:proper_for_sym}
    \kappa \Omega_{\tilde{M}}\left(\delta_1, \delta_2\right)=\int_{\tilde{M}} \delta_{[1} \Sigma^i \wedge \delta_{2]} K_i,  
  \end{equation}
  which is independent on the choice of the time-slices.
\end{theorem}
The proof of this theorem can be found in appendix \ref{app:detail}.

\subsection{The symplectic structutre of $\scri^+$}
In this subsection, we demonstrate that the dynamics on $\scri^+$ can be described by two Chern-Simons theories, in the sense of its symplectic structure.
We begin by rewriting the symplectic structure \eqref{eq:proper_for_sym} in terms of the Ashtekar-Barbero connection,
which is defined as (see etc. \cite{han2007fundamental,ashtekar2004background,rovelli2004quantum})
\begin{equation}
    A_a^i=\Gamma_a^i+\beta K_a^i,
    \end{equation}
where $\beta$ is the Ashtekar-Barbero-Immirzi parameter.
Then, the symplectic potential in eq. \eqref{eq:proper_for_sym} can be rewritten in terms of the  Ashtekar-Barbero variables as
\begin{equation}
    \begin{aligned}
        \kappa \Theta\left(\delta\right)
        =&\int_{\tilde{M}}\Sigma_i\wedge \delta K^i\\
        =&\frac{1}{\beta}\int_{\tilde{M}} \Sigma_i \wedge\delta
        \left(
            \Gamma^i
            +\beta K^i
        \right)
        -\frac{1}{\beta}\int_{\tilde{M}}\Sigma_i\wedge \delta  \Gamma^i\\
        =&\frac{1}{\beta}\int_{\tilde{M}}\Sigma_i \wedge\delta A^i
        +\frac{1}{\beta}\int_{\tilde{\Delta}}e_i\wedge \delta  e^i.\\
    \end{aligned}
\end{equation}
Here, we have used the identity for the spin connection introduced in \cite{thiemann2008modern}, which reads
\begin{equation}\label{eq:tran_delta_gama}
    \int_{\tilde{M}} \Sigma_i \wedge \delta \Gamma^i=-\int_{\tilde{\Delta}}e_i \wedge \delta e^i.
    \end{equation}
Consequently, the corresponding symplectic form reads
\begin{equation}\label{eq:sym_stru_ashte}
    \kappa \Omega_{\tilde M}\left(\delta_1, \delta_2\right)=\frac{1}{\beta} \int_{\tilde M} \delta_{[1} \Sigma^i \wedge \delta_{2]} A_i+\frac{1}{\beta} \int_{\tilde{\Delta}} \delta_{[1} e^i \wedge \delta_{2]} e_i,
    \end{equation}
where the boundary term defines the symplectic structure at a cross-section $\tilde{\Delta}$ of $\scri^+$
\begin{equation}\label{eq:sym_wih_trad}
    \kappa \Omega_{\tilde{\Delta}}\left(\delta_1, \delta_2\right)= \frac{1}{\beta} \int_{\tilde{\Delta}} \delta_{[1} e^i \wedge \delta_{2]} e_i.
\end{equation}
This boundary term is difficult to handle, as discussed in \cite{perez2011static}.
One approach to address this difficulty is to recast \eqref{eq:sym_wih_trad} into an equivalent formalism in  terms of two $SU(2)$ connections, defined as 
\begin{equation}
        A^i_\gamma :=\Gamma^i+\gamma e^i,\quad
        A^i_\sigma :=\Gamma^i+\sigma e^i.
\end{equation}
Here, $\gamma$ and $\sigma$ are two constants. 
With these definitions, eq. \eqref{eq:sym_wih_trad} can be rewritten in the following equivalent form
    \begin{equation}\label{eq:sym_with_chern_conne}
        \kappa \beta \Omega_{\tilde{\Delta}}\left(\delta_1, \delta_2\right)=\frac{1}{\gamma^2-\sigma^2} \int_{\tilde{\Delta}} \delta_{[1} A_\gamma^i \wedge \delta_{2]} A_{\gamma i}-\frac{1}{\gamma^2-\sigma^2} \int_{\tilde{\Delta}} \delta_{[1} A_\sigma^i \wedge \delta_{2]} A_{\sigma i}.
        \end{equation}
    To demonstrate this equivalence, we introduce the following lemma:
\begin{lemma}\label{lm:lemmar_for_curv}
    The curvatures of $A^i_\gamma$ and $A^i_\sigma$ take the following forms at $\scri^+$:
    \bea 
    \underset{\Leftarrow}{F^i}\left(A_\gamma\right)&=&\tilde\Psi_2 \underset{\Leftarrow}{\Sigma^i}+\frac{1}{2}\left(\gamma^2+c\right)  \underset{\Leftarrow}{\Sigma^i}\label{eq:cons_for_gamma},\\
    \underset{\Leftarrow}{F^i}\left(A_\sigma\right)&=&\tilde\Psi_2  \underset{\Leftarrow}{\Sigma^i}+\frac{1}{2}\left(\sigma^2+c\right)  \underset{\Leftarrow}{\Sigma^i}\label{eq:cons_for_sigma},
    \eea
    where $c$  is the scalar field introduced in theorem \ref{thm:thre_for_KK}.
\end{lemma}
The proof of this lemma can be found in \ref{app:detail}.
Using this result, we obtain the following theorem:
\begin{theorem}\label{th:syn_at_boun}
The symplectic structure \eqref{eq:sym_wih_trad} on $\tilde{\Delta}$ can be equivalent expressed in terms of the $SU(2)$ connections $A_\gamma$ and $A_\sigma$:
\begin{equation}
        \kappa \beta \Omega_{\tilde{\Delta}}\left(\delta_1, \delta_2\right)=\frac{1}{\gamma^2-\sigma^2} \int_{\tilde{\Delta}} \delta_{[1} A_\gamma^i \wedge \delta_{2]} A_{\gamma i}-\frac{1}{\gamma^2-\sigma^2} \int_{\tilde{\Delta}} \delta_{[1} A_\sigma^i \wedge \delta_{2]} A_{\sigma i}.
        \end{equation}
\end{theorem}
The proof can be found in appendix \ref{app:detail}.

Theorem \ref{th:syn_at_boun} indicates that the dynamics on $\tilde{\Delta}$ can be described by two $SU(2)$ Chern-Simons theories with levels
\begin{equation}
k_\gamma=-k_\sigma=\frac{8 \pi}{\left(\gamma^2-\sigma^2\right) \kappa \beta},
\end{equation}
see also \cite{engle2010black,perez2011static}.
This equivalence allows us to apply the quantization scheme of Chern–Simons theories introduced in \cite{witten1989quantum} to quantize $\scri^+$, as will be discussed in the next section.

\section{Quantization and entropy of $\mathscr{I}^+$} \label{sec:quanti_scri}
In this section, we first quantize $\mathscr{I}^+$ with the quantization scheme of Chern-Simons theory.
 Based on this result, we compute its entropy by counting the microstates.

A time-slice $\tilde{M}$ of the unphysical spacetime can be decomposed into a bulk region and a boundary $\tilde{\Delta}$, where $\tilde{\Delta}$ is the intersection between $\tilde{M}$ and $\scri^+$.
Following  the framework of LQG \cite{thiemann2008modern,rovelli2004quantum,ashtekar2004background,han2007fundamental}, we denote the Hilbert space of the bulk part associated with a fixed graph $\gamma$ as $\mathcal{H}^B_{\gamma}$.
The flux operator acting on $p\in \gamma \cap\tilde{\Delta}$, which is the intersection between $\gamma$ and $p\in \gamma \cap\tilde{\Delta}$ is defined as 
\begin{equation}\label{eq:con_for_bulk}
\epsilon^{ab}\hat\Sigma^i_{ab}
=16\pi G\hbar\beta\sum_{p\in \gamma \cap\mathcal{I}^+}\delta(x_{\text{bulk}},x)\hat J_{\text{bulk}}^i(p),
\end{equation}
where $\hat J_{\text{bulk}}^i(p)$ is the $SU(2)$ generator on the bulk, satisfying $\left[\hat J_{\text{bulk}}^i(p),\hat J_{\text{bulk}}^k(p)\right]=\epsilon^{ij}_{\,\,\,k}\hat J_{\text{bulk}}^k(p)$.
The area of $\tilde\Delta$ operator is defined as 
\begin{equation}
    \hat{A}_{\tilde\Delta}:=8\pi G\hbar\beta\sum_{p=1}^n\sqrt{\hat{J}^j_p\hat{J}^j_p},
\end{equation}
where the sum is over all punctures $p\in\gamma\cap\tilde{\Delta}$, and $\hat{J}^j_p$ are the corresponding $SU(2)$ operators at each puncture.
One can show that the eigenstates of $\hat{A}_{\tilde\Delta}$  are the spin network states on $\scri^+$, denoted as $\left|\left\{j_p, m_p\right\}_1^n ; \cdots\right\rangle$, where $j_p$ and $m_p$ are the spin and magnetic numbers respectively at each puncture.
The area operator acts on these  eigenstates as
\begin{equation}\label{eq: eigen_area}
    \hat{\tilde{A}}_{\tilde{\Delta}}\left|\left\{j_p, m_p\right\}_1^n ; \cdots\right\rangle
    =8\pi G\hbar\beta\sum_{p=1}^n\sqrt{j_p(j_p+1)}\left|\left\{j_p, m_p\right\}_1^n ; \cdots\right\rangle.
\end{equation}
We adopt Witten's quantization scheme for Chern-Simons theory \cite{witten1989quantum}.
At each puncture, we impose the following two constraints:
\begin{equation}
   \frac{k_\gamma}{4\pi}\epsilon^{ab} F^i_{ab}\left(A_\gamma\right)=J^i_\gamma(p),\qquad
    \frac{k_\sigma}{4\pi}\epsilon^{ab} F^i_{ab}\left(A_\sigma\right)=J^i_\sigma(p). 
\end{equation}
Here, $k_\gamma$ and $k_\sigma$ are the levels for these two Chern-Simons theories respectively, given by
\begin{equation}
    k_\gamma=-k_\sigma=\frac{8\pi}{\left(\gamma^2-\sigma^2\right)\kappa\beta}.
\end{equation}   
Recall \eqref{eq:cons_for_gamma} and \eqref{eq:cons_for_sigma}, eq. \eqref{eq:con_for_bulk} yields
\begin{equation}\label{eq: j_bound_to_bulk}
    J^i_\gamma(p)=\frac{2\tilde\Psi_2+\gamma^2+c}{\gamma^2-\sigma^2} J^i_{\text{bulk}}(p),\quad
  J^i_\sigma(p)=-\frac{2\tilde\Psi_2+\sigma^2+c}{\gamma^2-\sigma^2} J^i_{\text{bulk}}(p).
\end{equation}
From \eqref{eq: j_bound_to_bulk}, we obtain two constraints
\begin{equation}
    D^i(p)=J^i_{\text{bulk}}(p)+J^i_\gamma(p)+J^i_{\sigma}(p)=0,
\end{equation}
\begin{equation}
    C^i(p)=J^i_\gamma(p)-J^i_{\sigma}(p)+\frac{4\tilde\Psi_2+2c+\sigma^2+\gamma^2}{\gamma^2-\sigma^2}J^i_{\text{bulk}}(p)=0.
\end{equation}
We assign the pair of spins $\left(j_\gamma,\,j_\sigma \right)$ to each puncture.
As demonstrated in \cite{perez2011static, witten1989quantum}, the Hilbert space of the Chern-Simons theory is the intertwiner space of q-deformed $SU(2)$, denoted as $U_q\left(su(2)\right)$.
Therefore, the Hilbert space on $\tilde{\Delta}$ is 
\begin{equation}\label{eq:cs_hil}
    \mathscr{H}_k^{C S}\left(j_1^\gamma \cdots j_n^\gamma\right) \otimes \mathscr{H}_k^{C S}\left(j_1^\sigma \cdots j_n^\sigma\right) \subset \operatorname{Inv}\left(j_1^\gamma \otimes \cdots \otimes j_n^\gamma\right) \otimes \operatorname{Inv}\left(j_1^\sigma \otimes \cdots \otimes j_n^\sigma\right).
    \end{equation}
Here, $\operatorname{Inv}\left(j_i\right)$ is isomorphic to classical $SU(2)$ intertwiner space. 
Eq.\eqref{eq:cs_hil} introduces a cut-off $|k|/2$ to the spin $j$, with $k$ is the level of the Chern-simons theory.
As the result, we quantize $\tilde{\Delta}$, as a spacelike cross-section of $\scri^+$, via the Chern-Simons quantization scheme, with the Hilbert space given by $\mathscr{H}_k^{C S}\left(j_1^\gamma \cdots j_n^\gamma\right) \otimes \mathscr{H}_k^{C S}\left(j_1^\sigma \cdots j_n^\sigma\right)$.

With the above Hilbert space, the entropy of $\tilde{\Delta}$ is computed by counting the number of microstates.
According to eq. \eqref{eq: eigen_area}, the eigenvalue of the area operator on $\tilde{\Delta}$ is 
\begin{equation}\label{eq:eign_area_2}
A_{\tilde{\Delta}}=8\pi G\hbar\beta\sum_p\sqrt{(j_\gamma+j_\sigma)(j_\gamma+j_\sigma+1)}.
\end{equation}
Let $s\left(j_\gamma,j_\sigma\right)$ denotes the number of punctures carrying spin $j_\gamma+j_\sigma$.
With this configuration, eq. \eqref{eq:eign_area_2} becomes
\begin{equation}\label{eq:eign_area_3}
    A_{\tilde{\Delta}}=8\pi G\hbar\beta\sum_{j_\gamma.j_\sigma}^{\frac{|k|}{2}}s_j\left(j_\gamma,j_\sigma\right)\sqrt{(j_\gamma+j_\sigma)(j_\gamma+j_\sigma+1)}.
    \end{equation}
The number of the microstates is given by
\begin{equation}
d=\frac{\left[\sum_{j^\gamma,j^\sigma}^{\frac{|k|}{2}}s\left(j_\gamma,j_\sigma\right)\right]!}{\prod_{j^\gamma,j^\sigma}^{\frac{|k|}{2}}s\left(j_\gamma,j_\sigma\right)!}
\prod_{j^\gamma,j^\sigma}^{\frac{|k|}{2}}
\left(\left(2j^\gamma+1\right)\left(2j^\gamma+1\right)\right)
^{s\left(j_\gamma,j_\sigma\right)}.
\end{equation}
To relate $A_{\tilde{\Delta}}$ to the entropy of $\tilde{\Delta}$, we introduce 
\begin{equation}\label{eq:rel_micro_area}
    \delta \ln d=\lambda \delta A_{\tilde{\Delta}},
\end{equation}
with $\lambda$ is a Lagrangian multiple.
Following ref. \cite{perez2011static,ghosh2006counting}, we assume that $s\left(j_\gamma,j_\sigma\right)\gg 1$.
Then by Stirling's approximation, we have 
\begin{equation}
    \begin{aligned}
        \ln d
        =&\ln\left[\left[\sum_{j^\gamma,j^\sigma}^{\frac{|k|}{2}}s\left(j_\gamma,j_\sigma\right)\right]!\right]
        -\sum_{j^\gamma,j^\sigma}^{\frac{|k|}{2}}\ln\left[s\left(j_\gamma,j_\sigma\right)!\right]
        +\sum_{j^\gamma,j^\sigma}^{\frac{|k|}{2}}s\left(j_\gamma,j_\sigma\right)\ln\left[\left(2j^\gamma+1\right)\left(2j^\gamma+1\right)\right]\\
        \approx &
        \sum_{j^\gamma,j^\sigma}^{\frac{|k|}{2}}s\left(j_\gamma,j_\sigma\right)\ln\left[\sum_{j^\gamma,j^\sigma}^{\frac{|k|}{2}}s\left(j_\gamma,j_\sigma\right)\right]
        -\sum_{j^\gamma,j^\sigma}^{\frac{|k|}{2}}s\left(j_\gamma,j_\sigma\right)\ln\left[s\left(j_\gamma,j_\sigma\right)\right]\\
        &+\sum_{j^\gamma,j^\sigma}^{\frac{|k|}{2}}s\left(j_\gamma,j_\sigma\right)\ln\left[\left(2j^\gamma+1\right)\left(2j^\gamma+1\right)\right].
    \end{aligned}
\end{equation}
Therefore, with \eqref{eq:eign_area_2} and \eqref{eq:rel_micro_area}, we find
\begin{equation}
    \frac{s\left(j^\gamma, j^\sigma\right)}{\sum_{j^\gamma, j^\sigma}^{\frac{|k|}{2}} s\left(j^\gamma, j^\sigma\right)}=\left(2 j^\gamma+1\right)\left(2 j^\sigma+1\right) e^{-\lambda 8 \pi \beta \ell_p^2 \sqrt{\left(j^\gamma+j^\sigma\right)\left(j^\gamma+j^\sigma+1\right)}}.
\end{equation}
Thus, the entropy of $\scri^+$ reads
\begin{equation}
   \begin{aligned}
        S=&\ln d\\
        \approx&\sum_{j^\gamma,j^\sigma}^{\frac{|k|}{2}}s\left(j_\gamma,j_\sigma\right)\ln\left[\sum_{j^\gamma,j^\sigma}^{\frac{|k|}{2}}s\left(j_\gamma,j_\sigma\right)\right]
        -\sum_{j^\gamma,j^\sigma}^{\frac{|k|}{2}}s\left(j_\gamma,j_\sigma\right)\ln\left[s\left(j_\gamma,j_\sigma\right)\right]\\
        &+\sum_{j^\gamma,j^\sigma}^{\frac{|k|}{2}}s\left(j_\gamma,j_\sigma\right)\ln\left[\left(2j^\gamma+1\right)\left(2j^\gamma+1\right)\right]\\
        =&\sum_{j^\gamma,j^\sigma}^{\frac{|k|}{2}}s\left(j_\gamma,j_\sigma\right)
            \ln\left[\frac{\sum_{j^\gamma,j^\sigma}^{\frac{|k|}{2}}s\left(j_\gamma,j_\sigma\right)}
            {s}\left(2j^\gamma+1\right)\left(2j^\gamma+1\right)\right]\\
            =&\lambda A_{\tilde{\Delta}}.
   \end{aligned} 
\end{equation}
We conclude that the entropy of $\scri^+$ is proportional to the area of $\tilde{\Delta}$, which is consistent with the result obtained in the context of WIH.
\section{Conclusions and outlooks}\label{sec:con_out}
In this work, we show that the null infinity of the asymptotically flat spacetime ($\scri^+$) can be dynamically described  by a pair of $SU(2)$ Chern-Simons theories, based on the recent developments by Ashtekar and Speziale.
These developments demonstrate that $\scri^+$ is equivalent to a WIH under the condition of a divergence-free conformal frame, $\tilde{\nabla}_a\tilde{n}^a=0$ \cite{ashtekar2024null,ashtekar2024charge}.
Based on these developments, we show that the symplectic structure on $\tilde{\Delta}$, as a spacelike cross-section of $\scri^+$, is equivalent to the symplectic structure of two $SU(2)$ Chern-Simons theories with opposite levels.
We then quantize $\scri^+$ using the Chern-Simons quantization scheme
and compute its entropy by counting the number of the microstates.

This work represents a first step toward developing a quantum theory of $\scri^+$. 
However, the discussions on the gravitational radiation and BMS charges are absent from the current treatment.
One possible reason is the absence of the matter fields in the action eq. \eqref{eq:action}.
Henceforth, incorporating matter fields, e.g., electronic-magnetic field, coupled to the gravitational field would be a meaningful direction in the further research. 
This may lead to a formulation of the symplectic structure associated with BMS charges and gravitational radiation.
As mentioned in \cite{strominger2018lectures},  the generator of supertranslations is directly related to the gravitational memory effect. Therefore, our approach may offer useful insights into this phenomenon.

Another promising direction for future research is exploring the relationship between $\scri^+$ and Carrollian holography.
As previously discussed, Carrollian holography is built upon the isomorphism between the BMS symmetry in the bulk and the conformal Carroll group on the boundary ($\scri^+$).
A deeper investigation into $\scri^+$ could thus shed light on the structure of Carrollian holography.
Moreover, since the asymptotically flat spacetime can be regarded as a limiting case of asymptotically AdS spacetimes as $\Lambda\to0$,
these investigations might also contribute to a better understanding of the AdS/CFT correspondence.
\section*{Acknowledgment}
HT acknowledges the valuable discussions with Muxin Han, Xiangdong Zhang, Jinbo Yang and Rong-zhen Guo.
The authors acknowledge the support from the Hunan Provincial Natural Science Foundation of China (Grant No. 2022JJ30220).
\appendix
\section{details}\label{app:detail}
In this appendix, we provide the proofs of several important lemmas and theorems presented in the main text.
\begin{theorem}
    Under the gauge fixing introduced in sec. \ref{sec:wih_as_su2}, we have
\end{theorem}
\begin{equation}
\underset{\Leftarrow}{\boldsymbol{\omega}}^{21}=\underset{\Leftarrow}{\boldsymbol{\omega}}^{20},\quad
\underset{\Leftarrow}{\boldsymbol{\omega}}^{31}=\underset{\Leftarrow}{\boldsymbol{\omega}}^{30}
\end{equation}
\begin{proof}
    The expansion $\theta$ and the shear $\sigma$ of $\scri^+$ are given by
    \begin{equation}
        \theta\hat{=}\tilde m^a\bar{\tilde m}^b\tilde\nabla_a\tilde n_b,\quad
        \sigma\hat{=}\tilde m^a\tilde m^b\tilde\nabla_a\tilde n_b.
    \end{equation}

    Note that both $\rho$ and $\sigma$ vanish on the $\scri^+$, since $\scri^+$ is a WIH.
    Using \eqref{eq:def_of_connection}, we have
    \begin{equation}
        \begin{aligned}
        0 & \hat=\theta \hat=\frac{1}{2 \sqrt{2}}\tilde m^a\left(e_2^b-i e_3^b\right)\tilde \nabla_a\left(e_b^1- e_b^0\right)\\
        & \hat=\frac{1}{2 \sqrt{2}} \tilde m^a\left(\left(\omega_a^{21}-\omega_a^{20}\right)-i\left(\omega_a^{31}-\omega_a^{30}\right)\right).
        \end{aligned}
        \end{equation}
    Similarly, 
    \begin{equation}
        \begin{aligned}
        0 & \hat=\sigma \hat=\frac{1}{2 \sqrt{2}}\tilde m^a\left(e_2^b-i e_3^b\right)\tilde \nabla_a\left(e_b^1- e_b^0\right)\\
        & \hat=\frac{1}{2 \sqrt{2}}\tilde m^a\left(\left(\omega_a^{21}-\omega_a^{20}\right)-i\left(\omega_a^{31}-\omega_a^{30}\right)\right).
        \end{aligned}
        \end{equation}
    As a result, we conclude the following results
    \begin{equation}
        \underset{\Leftarrow}{\boldsymbol\omega}^{21}=\underset{\Leftarrow}{\boldsymbol\omega}^{20},\quad
        \underset{\Leftarrow}{\boldsymbol\omega}^{31}=\underset{\Leftarrow}{\boldsymbol\omega}^{30}.
        \end{equation}
\end{proof}
\begin{theorem}
    On $\scri^+$, we have
\begin{equation}\label{eq:thre_for_KK}
    \underset{\Leftarrow}{K}^j \wedge \underset{\Leftarrow}{K}^k \epsilon_{i j k}=c \underset{\Leftarrow}{\Sigma}^i,
    \end{equation}
    with $c$ is a scalar  field on $\hat\Delta$.
\end{theorem}

\begin{proof}
    To simplify the notations, all equations below are pulled back to $\scri^+$, and we suppress the "$\Leftarrow$" symbol.
    Since $e_1^{\,a}$ is normal to $\tilde\Delta$ by the gauge fixing,  $\Sigma^A=0$, where $A=2,\,3$.
    Therefore, only the $i=1$ component in eq. \eqref{eq:thre_for_KK} is non-trivial.
     Hence, it suffices to demonstrate
    \begin{equation}\label{eq:KK_1}
        K^A \wedge K^B \epsilon_{AB}=c\Sigma^1,
    \end{equation}
where the indices $A$ and $B$ are angular and take the values in $\{2,\,3\}$, and $\epsilon_{AB}=\epsilon_{1AB}$.
With the tetrad introduced in Sec. \ref{sec:wih_as_su2}, $K^A$ can be expanded as
\begin{equation}
    K^A=M^A_{\,B}e^B,
\end{equation}
with $M^A_{\,B}$ is the matrix of coefficients. Therefore, 
\begin{equation}
    \begin{aligned}
        &K^A\wedge K^B\epsilon_{AB}\\
    =&M^A_{\,C}M^B_{\,D}e^C\wedge e^D \epsilon_{AB}\\
    =&\det(M)e^A\wedge e^B\epsilon_{AB}.
    \end{aligned}
\end{equation}
Comparing this with eq. \eqref{eq:KK_1}, we have
\begin{equation}
    c=\det{M}.
\end{equation}
\end{proof}
To continue, we introduce the following lemma:
\begin{lemma}\label{lamma:2-1form_wedge}
    Given $\boldsymbol{A}$ a 2-form and $\boldsymbol{B}$ an 1-formon a 2-manifold, we have
     \begin{equation}
       \boldsymbol{A} v\lrcorner \boldsymbol{B}=-v\lrcorner (\boldsymbol{A})\wedge \boldsymbol{B},
    \end{equation}
    with $v$ a vector field tangents to this 2-manifold.
\end{lemma}
\begin{proof}
    On the 2-manifold, we have $\boldsymbol{A}\wedge \boldsymbol{B}=0$, Therefore
    \begin{equation}
        0=v\lrcorner(\boldsymbol{A}\wedge \boldsymbol{B})
        =\boldsymbol{A} v\lrcorner \boldsymbol{B}
        +v\lrcorner (\boldsymbol{A})\wedge \boldsymbol{B}.
    \end{equation}
\end{proof}
Then we have
\begin{theorem}
    The symplectic structure on a time-slice $\tilde{M}$ is
\begin{equation}
    \kappa \Omega_{\tilde{M}}\left(\delta_1, \delta_2\right)=\int_{\tilde{M}} \delta_{[1} \Sigma^i \wedge \delta_{2]} K_i,  
  \end{equation}
  which is independent on the choice of the time-slices.
\end{theorem}
\begin{proof}
To prove this theorem, we first demonstrate that there exists a function $\bar{\mu}$ on the phase space, such that
  \beq 
  \delta \mu=\delta \bar{\mu},
  \eeq
  where $\delta \mu$ is defined by eq. \eqref{eq:varia_mu} and $\delta\bar{\mu}$ takes the form:
  \begin{equation}
    \delta \tilde{\mu}=\int_{\tilde{M}_1} \Sigma_i \wedge \delta K^i-\int_{\tilde{M}_2} \Sigma_i \wedge \delta K^i.
    \end{equation}

From eq. \eqref{eq:varia_mu} and \eqref{eq:tran_delta_gama}, we find 
\begin{equation}
    \delta \mu-i \mathscr{D}(\delta)=\int_{\tilde M_1} \Sigma_i \wedge \delta K^i-\int_{\tilde M_2} \Sigma_i \wedge \delta K^i,
    \end{equation}
where
\begin{equation}
    \mathscr{D}(\delta)=\int_{\tilde{\Delta}_1-\tilde{\Delta}_2} e_i \wedge \delta e^i-\int_{\hat\scri^+} \Sigma_i^{+} \wedge \delta A^{+i}.
    \end{equation}
The $\mathscr{D}(\delta)$ can also be decomposed into a gauge transformation part and a diffeomorphism part ss $\mathscr{D}(\delta)=\mathscr{D}(\delta_\alpha)+\mathscr{D}(\delta_v)$.
Using \eqref{eq:gaue_transform}, we find the gauge transform part vanishes:
\begin{equation} 
    \begin{aligned}
        \mathscr{D}(\delta_\alpha)
        =&\int_{\tilde{\Delta}_1-\tilde{\Delta}_2}e_i\wedge\epsilon^i_{\,jk}\alpha^je^k
        +\int_{\hat\scri^+} \Sigma_i^{+} \wedge d_{A_{+}}\alpha^i\\
        =&\int_{\tilde{\Delta}_1-\tilde{\Delta}_2}e_i\wedge\epsilon^i_{\,jk}\alpha^je^k
        +\int_{\hat\scri^+} \Sigma_i^{+} \wedge\lt(d\alpha^i+\epsilon^i_{\,jk}A_{+}^j\alpha^k\rt)\\
        =&-\int_{\tilde{\Delta}_1-\tilde{\Delta}_2}\Sigma_i\alpha^i
        +\int_{\tilde{\Delta}_1-\tilde{\Delta}_2} \Sigma_i\alpha^i
        -\int_{\hat\scri^+}(d\lt(\Sigma^+_i\rt)\alpha^i
        -\epsilon_{ijk}\Sigma^{+i}\wedge A^{j}_+\alpha^k)\\
        =&-\int_{\scri^+}d_{A_+}\lt(\Sigma^+_i\rt)\alpha^i\\
        =&0.
    \end{aligned}
\end{equation}
Here, we have imposed the Gauss constraint in the last line.
According to \eqref{eq:proper_for_static}, \eqref{eq:diffeomorphism},  and  
lemma \ref{lamma:2-1form_wedge}, the diffeomorphism part vanishes:
\beq \label{eq:diff_sym_mid}
\begin{aligned}
    \mathscr{D}(\delta_v)
    =&\int_{\tilde{\Delta}_1-\tilde{\Delta}_2}
    e_i\wedge
    \left(
        v\lrcorner de^i
        +d\left(v\lrcorner e^i\right)
    \right)
    -\int_{\hat\scri^+}
    \Sigma_i^{+} \wedge
    \left(
        v\lrcorner F^i\left(A^+\right)
        +d_{A_+}\left(v\lrcorner A^{+i}\right)
    \right)\\
=&
    -\int_{\tilde{\Delta}_1-\tilde{\Delta}_2}
    v\lrcorner de^i\wedge e_i
    +\int_{\tilde{\Delta}_1-\tilde{\Delta}_2}e_i\wedge d\left(v\lrcorner e^i\right)
-\int_{\hat\scri^+}\Sigma_i^{+} \wedge d_{A_+}\left(v\lrcorner A^{+i}\right)\\
=&
2\int_{\tilde{\Delta}_1-\tilde{\Delta}_2}
v\lrcorner e_i de^i
-\int_{\tilde{\Delta}_1-\tilde{\Delta}_2}d\left(e_iv\lrcorner e^i\right)
-\int_{\hat\scri^+}d\left(\Sigma_i^{+}  v\lrcorner A^{+i}\right)
    +\int_{\hat\scri^+}d_{A_+}\left(\Sigma_i^{+}\right)  v\lrcorner A^{+i}\\
=&
-2\int_{\tilde{\Delta}_1-\tilde{\Delta}_2}
\epsilon_{ijk}v\lrcorner e_j e^k\wedge \Gamma^i
-\int_{\hat\scri^+}d\left(\Sigma_i^{+}  v\lrcorner A^{+i}\right)\\
    =&\int_{\tilde{\Delta}_1-\tilde{\Delta}_2}\Sigma_i  v\lrcorner \Gamma^{i}
    -\int_{\tilde{\Delta}_1-\tilde{\Delta}_2}\Sigma_i  v\lrcorner A^{+i}\\
    =&-i\int_{\tilde{\Delta}_1-\tilde{\Delta}_2}\Sigma_i  v\lrcorner K^{i}\\
    =&0.
\end{aligned}
\eeq
In  the second line, we have used the vector constraint $\Sigma_i\wedge\lrcorner\left[v\lrcorner F^i\lt(A^+\rt)\right]=0$ \cite{engle2010black}.
In the fourth line, we have employed  Cartan’s first structural equation $d_\Gamma e^{\,i}:=de^{\,i}+\epsilon^i_{\,jk}\Gamma^je^{k}=0$.
In the fifth line, we have used the following identity
\begin{equation}\label{eq:iden_e_sig}
    \begin{aligned}
        v\lrcorner\Sigma_i=&\epsilon_{ijk}v\lrcorner\lt(e^j\wedge e^k\rt)\\
        =&\epsilon_{ijk}v^b
        \lt(e_b^je_a^k
        -e_a^je_b^k
        \rt)\\
        =&2\epsilon_{ijk}v^b
        e_b^je_a^k\\
        =&2\epsilon_{ijk}v\lrcorner
        e^je^k
    \end{aligned}
\end{equation}
because $\delta_{[1}\wedge \delta_{2]}\bar{\mu}=0$, it follows that
\beq 
\int_{\tilde{M}_1}\delta_{[1}  \Sigma_i \wedge \delta_{2]}  K^i
=\int_{\tilde{M}_2} \delta_{[1}  \Sigma_i \wedge \delta_{2]}  K^i.
\eeq  
\end{proof}
\begin{lemma}
    The curvatures of $A^i_\gamma$ and $A^i_\sigma$ take the following forms at $\scri^+$:
    \bea 
    \underset{\Leftarrow}{F^i}\left(A_\gamma\right)&=&\tilde\Psi_2 \underset{\Leftarrow}{\Sigma^i}+\frac{1}{2}\left(\gamma^2+c\right)  \underset{\Leftarrow}{\Sigma^i},\\
    \underset{\Leftarrow}{F^i}\left(A_\sigma\right)&=&\tilde\Psi_2  \underset{\Leftarrow}{\Sigma^i}+\frac{1}{2}\left(\sigma^2+c\right)  \underset{\Leftarrow}{\Sigma^i},
    \eea
    where $c$ is the scalar field introduced in theorem \ref{thm:thre_for_KK}.
\end{lemma}
\begin{proof}
    In this proof, all quantities are pulled back $\scri^+$, and for simplicity, we omit the symbols $\Leftarrow$ in the notations.
    Throughout this paper, we considering a vacuum in the unphysical spacetime, where both the Ricci tensor and the Ricci scalar vanish. 
    Therefore, from eq. \eqref{eq:curvature_on_wih}, we obtain
    \begin{equation}
       F_{ab}^i(A^+)
        =\tilde\Psi_2\Sigma^i_{ab}.
    \end{equation}
   Thus, using eq. \eqref{eq:curva_forsel_dual}, we have
    \begin{equation}\label{eq:mid_for_psi2}
        \begin{aligned}
           \tilde \Psi_2\Sigma^i_{ab}
=&F^i_{ab}(A^+)
=dA_{a}^{+i}
+\frac{1}{2}\epsilon_{ijk}A_a^{+j}\wedge A_b^{+k}\\
=&d
(
    \Gamma_a^i
    +iK_a^i
)
+\frac{1}{2}\epsilon_{ijk}
(
    \Gamma_a^j
    +iK_a^j
)
\wedge 
(
    \Gamma_b^j
    +iK_b^k
)\\                    
=&
d\Gamma_a^i
+idK_a^i
+\frac{1}{2}\epsilon_{ijk}
\Gamma_a^j
\wedge 
\Gamma_b^j
+\frac{1}{2}\epsilon_{ijk}
iK_a^j
\wedge 
\Gamma_b^j
+\frac{1}{2}i\epsilon_{ijk}
\Gamma_a^j
\wedge 
K_b^k
-\frac{1}{2}\epsilon_{ijk}
K_a^j
\wedge 
K_b^k.
        \end{aligned}
    \end{equation}
Note that the one form field $\boldsymbol{\o}_{(\tilde{n})}=0$ indicates that $\im \Psi_2=0$. Hence, the imaginary part of eq. \eqref{eq:mid_for_psi2} vanishes, which leads to
\begin{equation}
    \begin{aligned}
        \tilde\Psi_2\Sigma^i_{ab}
=&d\Gamma_a^i
+\frac{1}{2}\epsilon_{ijk}
\Gamma_a^j
\wedge 
\Gamma_b^j
-\frac{1}{2}\epsilon_{ijk}
K_a^j
\wedge 
K_b^k
\\
=&d\Gamma_a^i
+\frac{1}{2}\epsilon_{ijk}
\Gamma_a^j
\wedge 
\Gamma_b^j
-\frac{1}{2}c\Sigma^i_{ab},
    \end{aligned}
\end{equation}
where we have used theorem \ref{thm:thre_for_KK}. 
Using Cartan's first structural equation, we find
\begin{equation}
    \begin{aligned}
        F^i(A_\gamma)
=&dA^i_\gamma
+\frac{1}{2}\epsilon_{ijk}
A^j_\gamma
\wedge
A^k_\gamma\\
=&
d
(
    \Gamma^i
    +\gamma e^i
)
+\frac{1}{2}\epsilon_{ijk}
(
    \Gamma^j
    +\gamma e^j
)
\wedge
(
    \Gamma^k
    +\gamma e^k
)\\
=&
d\Gamma^i
+\gamma d_\Gamma e^i
+\frac{1}{2}\epsilon_{ijk}
\Gamma^j
\wedge
\Gamma^k
+\frac{1}{2}\gamma^2\epsilon_{ijk}
 e^j
\wedge
 e^k
\\
=&
\tilde\Psi_2\Sigma^i
+\frac{1}{2}
(
    c
    +\gamma^2
)
\Sigma^i
    \end{aligned}
\end{equation}
With similar arguments for $A_\sigma^i$, we find 
\begin{equation}
    F^i\left(A_\sigma\right)=\tilde\Psi_2 \Sigma^i+\frac{1}{2}\left(c+\sigma^2\right) \Sigma^i\label{eq:curv_forAsigma}.
\end{equation}
\end{proof}
Then we have the following Theorem:
\begin{theorem}
The  symplectic structure \eqref{eq:sym_wih_trad} on $\tilde{\Delta}$ can be equivalent expressed in terms of the $SU(2)$ connections $A_\gamma$ and $A_\sigma$:
\begin{equation}\label{eq:sym_with_chern_conne_2}
        \kappa \beta \Omega_{\tilde{\Delta}}\left(\delta_1, \delta_2\right)=\frac{1}{\gamma^2-\sigma^2} \int_{\tilde{\Delta}} \delta_{[1} A_\gamma^i \wedge \delta_{2]} A_{\gamma i}-\frac{1}{\gamma^2-\sigma^2} \int_{\tilde{\Delta}} \delta_{[1} A_\sigma^i \wedge \delta_{2]} A_{\sigma i}.
        \end{equation}
\end{theorem}
\begin{proof}
    To prove eq. \eqref{eq:sym_with_chern_conne_2},
    It is sufficient to show that the following symplectic potential is closed:
    \begin{equation}
        \Theta_0(\delta) \equiv \int_{\tilde{\Delta}} e_i \wedge \delta e^i-\frac{1}{\gamma^2-\sigma^2} \int_{\tilde{\Delta}}\left(A_{\gamma i} \wedge \delta A_\gamma^i-A_{\sigma i} \wedge \delta A_\sigma^i\right),
        \end{equation}
  which implies 
    \begin{equation}
        \int_{\tilde{\Delta}}\delta e_i \wedge \delta e^i
        =\frac{1}{\gamma^2-\sigma^2} \int_{\tilde{\Delta}}\left(\delta A_{\gamma i} \wedge \delta A_\gamma^i-\delta A_{\sigma i} \wedge \delta A_\sigma^i\right).
    \end{equation}
We define the exterior variation of the potential as
    \begin{equation}
        \mathfrak{d} \Theta_0\left(\delta_1, \delta_2\right):=\delta_1\left(\Theta_0\left(\delta_2\right)\right)-\delta_2\left(\Theta_0\left(\delta_1\right)\right).
        \end{equation}
        With the gauge transformations given in eq. \eqref{eq:gaue_transform}, we find 
        \begin{equation}
            \begin{aligned}
                &\mathfrak{d}\Theta_0\left(\delta, \delta_\alpha\right)\\
                =&\int_{\tilde{\Delta}}2\delta e^i\wedge\delta_\alpha e_i
                -\frac{2}{\gamma^2-\sigma^2}\int_{\tilde{\Delta}}
                \left(
                    \delta A_{\gamma i}\wedge\delta_\alpha A^i_\gamma
                    -\delta A_{\sigma i}\wedge\delta_\alpha A^i_\sigma
                \right)\\
                =&\int_{\tilde{\Delta}}2\delta e^i\wedge\epsilon_{ijk}\alpha^je^k
                +\frac{2}{\gamma^2-\sigma^2}\int_{\tilde{\Delta}}
                \left(
                    \delta A_{\gamma i}\wedge 
                    \left(
                        d\alpha^i
                        +\epsilon_{ijk} A_{\gamma j}\alpha^k
                    \right)
                    -\delta A_{\sigma i}\wedge 
                    \left(
                        d\alpha^i
                        +\epsilon_{ijk} A_{\sigma j}\alpha^k
                    \right)
                    \right)\\
                    =&\int_{\tilde{\Delta}}2\delta e^i\wedge\epsilon_{ijk}\alpha^je^k
                    +\frac{2}{\gamma^2-\sigma^2}\int_{\tilde{\Delta}}
                    \left(
                        \delta d\left(A_{\gamma i}\right) \wedge \alpha^i
                        +\frac{1}{2}\epsilon_{ijk} \delta (A_{\gamma i}\wedge 
                        A_{\gamma j})\alpha^k\right)\\
                    &-\frac{2}{\gamma^2-\sigma^2}\int_{\tilde{\Delta}}
                    \left(
                        \delta d\left(A_{\sigma i}\right) \alpha^i
                        +\frac{1}{2}\epsilon_{ijk} \delta (A_{\sigma i}\wedge 
                        A_{\sigma j})\alpha^k
                \right)\\
                =&\int_{\tilde{\Delta}}-\delta \Sigma^i\alpha_i
                +\frac{2}{\gamma^2-\sigma^2}\int_{\tilde{\Delta}}\left(
                    \delta F\left(A_\gamma\right)\alpha^i
                    -\delta F\left(A_\sigma\right) \alpha^i
                \right).
            \end{aligned}
        \end{equation}
        Then using lemma \ref{lm:lemmar_for_curv}, we obtain 
        \begin{equation}
            \mathfrak{d}\Theta_0\left(\delta, \delta_\alpha\right)=0.
        \end{equation}
Using the diffeomorphism given in eq. \eqref{eq:diffeomorphism}, we have
\begin{equation}\label{eq:boun_sym_in_conne}
    \mathfrak{d}\Theta_0\left(\delta, \delta_v\right)
                =\int_{\tilde{\Delta}}2\delta e^i\wedge\delta_v e_i
                -\frac{2}{\gamma^2
                -\sigma^2}\int_{\tilde{\Delta}}
                \left(
                    \delta A_{\gamma i}\wedge\delta_v A^i_\gamma
                    -\delta A_{\sigma i}\wedge\delta_v A^i_\sigma
                \right).
\end{equation}
With eq. \eqref{eq:iden_e_sig} and Cartan’s first structural equation,  we find  the  contribution of the tetrads in eq. \eqref{eq:boun_sym_in_conne} as
\begin{equation}
    \begin{aligned}
        &2\int_{\tilde{\Delta}}\delta e_i\wedge \delta_v e_i\\
=&2\int_{\tilde{\Delta}}\delta e_i\wedge
\lt(
    v\lrcorner de_i
    +d\lt(v\lrcorner e_i\rt)
\rt)\\
=&
2\int_{\tilde{\Delta}}v\lrcorner\delta e_i\wedge
 de_i
+2\int_{\tilde{\Delta}}\delta de_i\wedge
v\lrcorner e_i\\
=&2\int_{\tilde{\Delta}}\delta \lt(de_i\wedge
v\lrcorner e_i\rt)\\
=&-2\int_{\tilde{\Delta}}\delta \lt(
    \epsilon_{ijk}
\Gamma^j\wedge e^k    \wedge
v\lrcorner e_i\rt)\\
=&2\int_{\tilde{\Delta}}\delta \lt(
    \epsilon_{ijk}
\Gamma^i\wedge e^k    \wedge
v\lrcorner e_j\rt)\\
=&\int_{\tilde{\Delta}}\delta \lt(\Gamma^i
    \wedge  
    v\lrcorner 
    \Sigma^i\rt).
    \end{aligned}
\end{equation}
By Cartan's magic equation,  the diffeomorphism generated by the tangent vector $v$ on the connection $A^i_\sigma$ is given by
\begin{equation}
    \begin{aligned}
        \delta_v A_\gamma^i:=\mathcal{L}_{\vec{v}}A_\gamma^i
=&d\left(v\lrcorner A^i_\gamma\right)
+v\lrcorner dA^i_\gamma\\
=&d\left(v\lrcorner A^i_\gamma\right)
+v\lrcorner F^i(A_\gamma)
-\frac{1}{2}\epsilon_{ijk}v\lrcorner\left(A^j_\gamma\wedge A^k_\gamma\right)\\
=&d_{A_\gamma}\left(v\lrcorner A^i_\gamma\right)
+v\lrcorner F^i\lt(A_\gamma\rt).
    \end{aligned}
\end{equation}
Similarly, we have
\begin{equation}
    \delta_v A_\sigma^i
    =d_{A_\sigma}\lt(v\lrcorner A^i_\sigma\rt)
    +v\lrcorner F^i\lt(A_\sigma\rt).
\end{equation}
Then we find
\begin{equation}\label{eq:_sym_mid}
    \begin{aligned}
       & \int_{\tilde{\Delta}}\delta A_{\gamma i}
\wedge 
\delta_v A_\gamma^i
-\int_{\tilde{\Delta}}\delta A_{\sigma i}
\wedge 
\delta_v A_\sigma^i\\
=&
\int_{\tilde{\Delta}}\delta A_{\gamma i}
\wedge  
v\lrcorner F^i\lt(A_\gamma\rt)
+\int_{\tilde{\Delta}}\delta A_{\gamma i}
\wedge  
d_{A_{\gamma}}\lt(v\lrcorner A^i_\gamma\rt)\\
&-\int_{\tilde{\Delta}}\delta A_{\sigma i}
\wedge  
v\lrcorner F^i(A_\sigma)
-\int_{\tilde{\Delta}}\delta A_{\sigma i}
\wedge  
d_{A_{\sigma}}(v\lrcorner A^i_\sigma)\\
=&
\int_{\tilde{\Delta}}\delta 
(
    \Gamma^i
    +\gamma e^i
)
\wedge  
v\lrcorner F^i(A_\gamma)
+\int_{\tilde{\Delta}}\delta A_{\gamma i}
\wedge  
d_{A_{\gamma}}(v\lrcorner A^i_\gamma)\\
&-\int_{\tilde{\Delta}}\delta 
(
    \Gamma^i
    +\sigma e^i
)
\wedge  
v\lrcorner F^i(A_\sigma)
-\int_{\tilde{\Delta}}\delta A_{\sigma i}
\wedge  
d_{A_{\sigma}}(v\lrcorner A^i_\sigma)\\
=&
\int_{\tilde{\Delta}}\delta \Gamma^i
\wedge  
v\lrcorner 
(
    F^i(A_\gamma)
    -F^i(A_\sigma)
)
+ 
\int_{\tilde{\Delta}}\gamma\delta e^i
\wedge  
v\lrcorner F^i(A_\gamma)
+\int_{\tilde{\Delta}}\delta A_{\gamma i}
\wedge  
d_{A_{\gamma}}(v\lrcorner A^i_\gamma)\\
&-\int_{\tilde{\Delta}}\sigma\delta e^i
\wedge  
v\lrcorner F^i(A_\sigma)
-\int_{\tilde{\Delta}}\delta A_{\sigma i}
\wedge  
d_{A_{\sigma}}(v\lrcorner A^i_\sigma)
\\
=&\int_{\tilde{\Delta}}
\delta \Gamma^i
\wedge  
v\lrcorner 
(
    F^i(A_\gamma)
    -F^i(A_\sigma)
)
+ 
\int_{\tilde{\Delta}}\gamma\delta e^i
\wedge  
v\lrcorner F^i(A_\gamma)
+\int_{\tilde{\Delta}}\delta F_i\lt(A_{\gamma }\rt)
v\lrcorner A^i_\gamma\\
&-\int_{\tilde{\Delta}}\sigma\delta e^i
\wedge  
v\lrcorner F^i(A_\sigma)
-\int_{\tilde{\Delta}}\delta F_i(A_{\sigma}) 
v\lrcorner A^i_\sigma\\
=&\int_{\hat{\Delta}}
\delta \Gamma^i
\wedge  
v\lrcorner 
(
    F^i(A_\gamma)
    -F^i(A_\sigma)
)
+ 
\int_{\tilde{\Delta}}\gamma\delta e^i
\wedge  
v\lrcorner F^i(A_\gamma)
+\int_{\tilde{\Delta}}\delta F_i\lt(A_{\gamma }\rt)
v\lrcorner \left(\Gamma^i+\gamma e^i\right)\\
&-\int_{\tilde{\Delta}}\sigma\delta e^i
\wedge  
v\lrcorner F^i(A_\sigma)
-\int_{\tilde{\Delta}}\delta F_i(A_{\sigma}) 
v\lrcorner \left(\Gamma^i+\sigma e^i\right)\\
=&\int_{\tilde{\Delta}}
\delta \Gamma^i
\wedge  
v\lrcorner 
(
    F^i(A_\gamma)
    -F^i(A_\sigma)
)
+ 
\int_{\tilde{\Delta}}\gamma\delta\lt( e^i
\wedge  
v\lrcorner F^i(A_\gamma)\rt)
+\int_{\tilde{\Delta}}\delta F_i\lt(A_{\gamma }\rt)
v\lrcorner \Gamma^i\\
&-\int_{\tilde{\Delta}}\sigma\delta\lt( e^i
\wedge  
v\lrcorner F^i(A_\sigma)\rt)
-\int_{\tilde{\Delta}}\delta F_i(A_{\sigma}) 
v\lrcorner \Gamma^i.
    \end{aligned}
\end{equation}
In the last line, we have applied lamma \ref{lamma:2-1form_wedge}.
Then according to lemma \ref{lm:lemmar_for_curv}, we have 
\begin{eqnarray}
    e^i\wedge v\lrcorner F(A_{\gamma })_i
    =
    (
        \Psi_2+
        \frac{1}{2}
        (
            \gamma^2
            +c
        )
    )
    e^i\wedge v\lrcorner \Sigma_i,\nonumber\\
    e^i\wedge v\lrcorner F(A_{\sigma })_i
    =
    (
        \Psi_2+
        \frac{1}{2}
        (
            \sigma^2
            +c
        )
    )
    e^i\wedge v\lrcorner \Sigma_i.
\end{eqnarray}
According to Theorem \ref{thm:thre_for_KK}, we find
\begin{equation}
    \begin{aligned}
        e^i\wedge v\lrcorner \Sigma_i
=&e^i_a\wedge v^b \Sigma_{ibc}\\
=&c^{-1}v^b \epsilon_{ijk}K^j_be^i_a\wedge K^k_c.\\
    \end{aligned}
\end{equation}
Then with the identity \eqref{eq:proper_for_static}  given by the fact that $\im \Psi_2=0$, we find
\begin{equation}
    \begin{aligned}
        0=&v^b K^i_b\Sigma^i_{ac}e^{c}_k\\
=&v^b K^i_b\epsilon_{ijl}e_a^j\wedge e^l_c e^{c}_k\\
=&v^b K^i_b\epsilon_{ijl}(e_a^j e^l_c-e_c^j e^l_a) e^{c}_k\\
=&2v^b K^i_b\epsilon_{ijl}e_a^j\delta^l_k\\
=&2v^b K^i_b\epsilon_{ijk}e_a^j,
    \end{aligned}
\end{equation}
which indicates
\begin{equation}
    e^i\wedge v\lrcorner \Sigma_i
=cv^b \epsilon_{ijk}K^j_be^i_a\wedge K^k_c
=0.
\end{equation}
Therefore, we find $\int_{\tilde{\Delta}}\gamma\delta\lt( e^i\wedge v\lrcorner F^i(A_\gamma)\rt)=\int_{\tilde{\Delta}}\sigma\delta\lt( e^i\wedge v\lrcorner F^i(A_\sigma)\rt)=0$ in eq. \eqref{eq:_sym_mid}.
So we obtain
\begin{equation}
    \begin{aligned}
        &\int_{\tilde{\Delta}}\delta A_{\gamma i}
\wedge 
\delta_v A_\gamma^i
-\int_{\tilde{\Delta}}\delta A_{\sigma i}
\wedge 
\delta_v A_\sigma^i\\
=&\int_{\tilde{\Delta}}
\delta \Gamma^i
\wedge  
v\lrcorner 
(
    F^i(A_\gamma)
    -F^i(A_\sigma)
)
+\int_{\tilde{\Delta}}
\delta
(
    F_i(A_{\gamma })
    -F_i(A_{\sigma })
)
v\lrcorner \Gamma^i.
    \end{aligned}
\end{equation}
where  
\begin{equation}
    \begin{aligned}
        \underset{\Leftarrow}{F^i}(A_{\gamma })
-\underset{\Leftarrow}{F^i}(A_{\sigma })
=\frac{\gamma^2-\sigma^2}{2}\underset{\Leftarrow}{\Sigma^i}.
    \end{aligned}
\end{equation}
Then by lamma \ref{lamma:2-1form_wedge}, we find
\begin{equation}
    \begin{aligned}
        \int_{\tilde{\Delta}}\frac{2}{\gamma^2-\sigma^2}
(
    \delta A_{\gamma i}
    \wedge 
    \delta_v A_\gamma^i
    -\delta A_{\sigma i}
    \wedge 
    \delta_v A_\sigma^i
)
=&\int_{\tilde{\Delta}}
(
    \delta \Gamma^i
    \wedge  
    v\lrcorner 
    \Sigma^i
    +
    \delta
    \Sigma^i
    v\lrcorner \Gamma^i
)\\
=&\int_{\tilde{\Delta}}
\delta (\Gamma^i
    \wedge  
    v\lrcorner 
    \Sigma^i).
    \end{aligned}
\end{equation}
As the result, we find
\begin{equation}
    \begin{aligned}
        \mathfrak{d}\Theta_0\left(\delta, \delta_v\right)
                =&\int_{\tilde{\Delta}}2\delta e^i\wedge\delta_v e_i
                -\frac{2}{\gamma^2-\sigma^2}\int_{\tilde{\Delta}}
                \left(
                    \delta A_{\gamma i}\wedge\delta_v A^i_\gamma
                    -\delta A_{\sigma i}\wedge\delta_v A^i_\sigma
                \right)\\
                =&\int_{\tilde{\Delta}}\delta (\Gamma^i
                \wedge  
                v\lrcorner 
                \Sigma^i)
                -\delta (\Gamma^i
                \wedge  
                v\lrcorner 
                \Sigma^i)\\
                =&0.
    \end{aligned}
\end{equation}
Henceforth, we get 
\begin{equation}
            \mathfrak{d}\Theta_0\left(\delta, \delta_v\right)=0.
        \end{equation}
Finally, we have demonstrated 
\begin{equation}
\delta \Theta_0(\delta)=0,
\end{equation}
which indicates the symplectic form at boundary  can be expressed as 
\begin{equation}
            \kappa \beta \Omega_{\tilde{\Delta}}\left(\delta_1, \delta_2\right)=\frac{1}{\gamma^2-\sigma^2} \int_{\tilde{\Delta}} \delta_{[1} A_\gamma^i \wedge \delta_{2]} A_{\gamma i}-\frac{1}{\gamma^2-\sigma^2} \int_{\tilde{\Delta}} \delta_{[1} A_\sigma^i \wedge \delta_{2]} A_{\sigma i},
\end{equation}
\end{proof}
\section{The closure of the symplectic structure}\label{app:clo_J}
In this appendix, we demonstrate that the symplectic current $J$ is closed by imposing the equations of motion, i.e., the on-shell condition.
This is a universal result, which is independent of the specific theory under consideration.

Consider an $n$-dimensional spacetime $\left(\mathscr{M},g_{ab}\right)$. The action of the system is given by
\begin{equation}\label{eq:gen_act}
S=\int_{\mathscr{M}}\boldsymbol{L}(\phi^i, \partial_a\phi^i,....),
\end{equation}
where $\boldsymbol{L}$ is the Lagrangian density $n$-form, and $\phi^i$ collectively  denotes the dynamical fields of the system, including the metric field $g_{ab}$ and possible matter fields.
In general, the Lagrangian density $\boldsymbol{L}$ includes any higher-order derivate of $\phi^i$.

The variation of $\boldsymbol{L}$ in eq. \eqref{eq:gen_act}, combined with integration by part, yields
\begin{equation}\label{eq:del_L}
    \delta \boldsymbol{L}
    =\boldsymbol{E}_i\delta \phi^i+d\theta(\phi^i,\delta\phi^i).
\end{equation}
Imposing the on-shell condition $\boldsymbol{E}_i=0$, eq. \eqref{eq:del_L} reduces to
\begin{equation}\label{eq:del_L_onshell}
    \delta \boldsymbol{L}
    =d\theta(\phi^i,\delta\phi^i).
\end{equation}
The symplectic current is defined as 
\begin{equation}
    J(\phi;\delta_1\phi,\delta_2\phi)
    :=\delta_1\theta(\phi,\delta_2\phi)
    -\delta_2\theta(\phi,\delta_1\phi).
\end{equation}
Thus, eq. \eqref{eq:del_L_onshell} indicates
\begin{equation}
    dJ=\delta_1\delta_2\boldsymbol{L}
    -\delta_2\delta_1\boldsymbol{L}
    =0.
\end{equation}
\bibliographystyle{jhep} 
\bibliography{myref} 
\end{document}